\documentclass{article}
\usepackage{graphicx,amsmath,amsthm,dsfont,amsfonts,xcolor}
\usepackage{bm, booktabs, subcaption}
\usepackage{natbib}
\usepackage{microtype}
\usepackage[noline,ruled,vlined]{algorithm2e}
\usepackage{comment}
\usepackage{float}
\usepackage{amssymb}
\usepackage{amsmath}
\usepackage{multirow}
\usepackage{threeparttable}
\usepackage[section]{placeins}
\usepackage[toc,page]{appendix}

\newcommand{\Var}{\operatorname{Var}}

\def\d{\,\mathrm{d}}

\newcommand{\VaR}{\mathrm{VaR}}                      
\newcommand{\ES}{\mathrm{ES}}                        
\newcommand{\E}{\mathbb{E}}                          
\newcommand{\R}{\mathbb{R}}

\newcommand{\one}{\mathbf{1}}                        
\newcommand{\id}{\mathds{1}}                        
\renewcommand{\ge}{\geqslant}
\renewcommand{\le}{\leqslant}
\renewcommand{\geq}{\geqslant}
\renewcommand{\leq}{\leqslant}
\renewcommand{\epsilon}{\varepsilon}

\usepackage[colorlinks=true,
            linkcolor=blue,
            citecolor=blue,
            urlcolor=blue]{hyperref}

\theoremstyle{plain}
\newtheorem{theorem}{Theorem}
\newtheorem{corollary}{Corollary}
\newtheorem{lemma}{Lemma}
\newtheorem{proposition}{Proposition}

\theoremstyle{definition}

\newtheorem{example}{Example}

\theoremstyle{remark}
\newtheorem{remark}{Remark}

\renewcommand{\cite}{\citet}

\DeclareMathOperator*{\argmin}{arg\,min}

\renewcommand{\d}{\mathrm{d}}

 \usepackage{setspace}

\title{\bf Adaptive Window Selection for Financial Risk Forecasting}

\author{
Yinhuan Li,\quad Chenxin Lyu,\quad and Ruodu Wang\thanks{Corresponding author. Email: ruodu.wang@uwaterloo.ca}\\[0.4em]
\small Department of Statistics and Actuarial Science, University of Waterloo, Canada
}

\date{}
\begin{document}

\maketitle
\begin{abstract}
Risk forecasts in financial regulation and internal management are calculated through historical data. The unknown structural changes of financial data pose a substantial challenge in selecting an appropriate look-back window for risk modeling and forecasting. We develop a data-driven online learning method, called the bootstrap-based adaptive window selection (BAWS), that adaptively determines the window size in a sequential manner. 
A central component of  BAWS is to compare the realized scores against a data-dependent threshold based on the bootstrap method. We provide an asymptotic justification for the bootstrap threshold, covering non-smooth
scores such as the VaR check loss and the joint VaR--ES score, with an extension to
stationary weakly dependent data via the moving block bootstrap. A single-break analysis
further shows that BAWS rejects overlong windows crossing sufficiently large breaks. The proposed method is applicable to the forecasting of risk measures that are elicitable individually or jointly, such as the Value-at-Risk (VaR) and the pair of VaR and the corresponding Expected Shortfall.  Through simulation studies and an empirical analysis, we demonstrate that BAWS  often improves upon the standard rolling window approach and the recently developed method of stability-based adaptive window selection, especially when there are structural changes in the data-generating process.
\end{abstract}

\noindent%
\vfill
\section{Introduction}
Forecasting risk measures, such as Value-at-Risk (VaR) and Expected Shortfall (ES), is central to financial regulation and internal risk management. These forecasts are typically estimated from historical financial data, making the choice of look-back window crucial. In reality, financial markets evolve with macroeconomic conditions and frequently experience unknown structural breaks, rendering the data highly non-stationary. For example, systemic shocks such as the 2008 global financial crisis (GFC) and the COVID-19 pandemic (COVID) triggered substantial changes in market dynamics and led to abrupt shifts in the loss distribution \citep{huber2021market}. 
Determining an appropriate estimation window under such changing market conditions therefore remains a challenging problem.

A rule-of-thumb for addressing this issue is the rolling window approach, which applies a fixed-length window that moves forward over time to update forecasts \citep{rapach2013forecasting}. 
However, the choice of window size is typically ad hoc. Regulatory frameworks such as Basel II/III and the FRTB provide only minimum requirements, such as using at least 250 observations and an additional 12-month stress period for ES forecasts. 
As a result, researchers and practitioners often rely on heuristic choices: \cite{demiguel2009optimal}  and \cite{capponi2022systemic}  employ a five- or ten-year window for portfolio selection, while \cite{hoga2023monitoring} and \cite{wang2025backtesting} adopt 250- or 500-day windows  for risk forecasting and backtesting. 

Nevertheless, the rolling window approach has important limitations because it implicitly assumes stationarity within each fixed window. The first limitation is that forecast performance is highly sensitive to the choice of window size \citep{rossi2012out}. This sensitivity reflects the bias-variance trade-off in fixed-window estimation: a long window generally yields low-variance forecasts  but incurs a high bias when structural breaks occur, while a short window reduces the  bias but produces more volatile forecasts.  The second limitation arises from the fixed-window mechanism itself, which cannot automatically adapt to evolving market conditions. Reliable risk forecasts must respond rapidly to market changes, but a fixed window tends to dilute extreme losses by averaging over pre-shift data, which may lead to underestimation of risk. These limitations imply that no single window size can perform well across various market conditions, motivating the development of window selection strategies more resilient to distributional changes.

A growing literature studies the role of estimation window choice in forecasting when economic time series, such as financial returns, exhibit instabilities or structural breaks; see, e.g., \cite{rapach2013forecasting}, \cite{rossi2021forecasting}, and the references therein. Related adaptive-window ideas have also been studied for model assessment and selection under temporal distribution shift; see \citet{han2024model}. Existing approaches to window selection and adaptive forecasting are typically motivated by parameter instability in predictive models. \cite{pesaran2007selection} show that pre-break observations may remain informative for parameter estimation in the presence of one or multiple structural breaks. They propose two approaches exploiting the bias-variance trade-off: selecting a single estimation window via cross-validation, or combining forecasts obtained from different windows. \cite{clark2009improving} further consider combinations of recursive and rolling forecasts, deriving optimal time-varying combination weights under abrupt parameter changes. Likewise, \cite{inoue2017rolling} study the selection of the optimal rolling window size for linear predictive models with time-varying coefficients and \cite{feng2025forecasting} 
investigate the optimal rolling window size by comparing the prediction performance of volatility
under various window sizes via the Diebold-Mariano test \citep{diebold1995comparing}.
Beyond window selection, several studies have explored alternative ways to adaptively utilize historical information to improve forecasting performance. For instance, \cite{pesaran2013optimal} propose assigning optimal weights to past observations, \cite{giraitis2013adaptive} focus on selecting an optimal rate of downweighting older data, and \cite{wang2021forecasting} investigate time-varying weighting schemes for historical observations. Most, if not all, of these approaches are developed for linear predictive regression  and aim to minimize mean squared forecast error, thereby targeting conditional mean forecasts. Therefore, they are not directly applicable to risk forecasting, where the object of interest is the tail risk measure rather than the conditional mean.

Complementary to these studies, we develop a model-free adaptive window selection
strategy for risk forecasting with elicitable scoring losses. Our method builds upon the
stability-based adaptive window selection (SAWS) framework of
\citet{huang2025stability}, which compares empirical losses across candidate windows
using deterministic thresholds. We complement this framework by constructing  a data-dependent threshold constructed by the bootstrap method rather than the deterministic threshold.   In this way, the threshold is
calibrated from the observed data while retaining the stability principle of SAWS.
The resulting bootstrap-based adaptive window selection (BAWS)
procedure uses a single threshold-level parameter and provides a flexible window-selection rule for risk forecasting.

We also study the theoretical behavior of BAWS.
Under a non-smooth M-estimation setting,  the bootstrap excess loss has the same asymptotic distribution as the population excess loss, which motivates the bootstrap threshold in window comparisons. This justification covers the VaR check loss and the joint VaR--ES  score. An extension to stationary weakly dependent data using the moving block bootstrap is given in the Supplementary Material. A single-break mean-estimation case further shows that an overlong window crossing a sufficiently large break is rejected with probability tending to one.

Our method compares forecasts obtained under a large window with those from smaller candidate windows using empirical scores, which requires scoring functions that are consistent for risk measures of interest. Elicitability provides this foundation:
a risk measure is  elicitable if it can be represented as a minimizer of an expected scoring function; see  \cite{gneiting2011making}, \cite{fissler2016higher} and \cite{fissler2024elicitability}.  This property allows us to apply the proposed adaptive window-selection framework to risk forecasting.    VaR (quantile) is a  prototypical example of elicitable risk measures, and although ES on its own is not elicitable, the pair (VaR, ES) is jointly elicitable \citep{acerbi2014back,fissler2015expected}. 
Some other popular statistical quantities that are elicitable include the mean, the (mean, variance) pair, and the expectile \citep{newey1987asymmetric}. 
For relevance in risk forecasts, we focus on the adaptive window selection for VaR or ES forecasts, leveraging their (joint) elicitability to construct data-driven criteria for comparing competing candidate windows.

Through  three simulation studies and an empirical analysis, we demonstrate that the proposed BAWS approach generally shows superior out-of-sample performance, though not uniformly in every scenario. Across three non-stationary settings, including discrete breaks, smooth and continuous changes, and time-varying volatility, our adaptive window approach delivers lower cumulative risk and forecast loss than fixed-window and full-window approaches, and often performs competitively with or better than SAWS. For instance, under the GARCH volatility-shift design, our method attains the lowest MSE, cumulative risk, and forecast loss among all competing procedures, highlighting its favorable bias–variance trade-off and swift reaction to regime changes. In the empirical analysis, BAWS and SAWS achieve lower cumulative forecast losses and respond more promptly to extreme events, such as the 2008 financial crisis and the COVID-19 pandemic, whereas fixed and full windows react with substantial delays.

The remainder of the paper is structured as follows. 
Section~\ref{sec:problem} states the research problem and introduces the bootstrap-based threshold. 
Section~\ref{sec:theory} presents  theoretical results and a two-regime illustration. 
Section~\ref{sec:risk_fore} introduces elicitability-based loss functions for risk forecasting. 
Sections~\ref{sec:Sim} and~\ref{sec:emp} report the simulation and empirical results, respectively. 
Section~\ref{sec:con} concludes, and proofs and additional theoretical results are relegated to the Supplementary Material.

\section{Methodology}\label{sec:problem}
\subsection{Problem setup}
Let
$\Theta \subseteq \mathbb{R}^d$ be a parameter space. We consider a sequence of random vectors \(\{\bm X_t\}_{t\ge1}\), where each \(\bm X_t\) takes values in \(\mathbb R^p\) and follows an unknown distribution \(\mathbb P_t\). 
The time horizon may be finite or infinite.

Let $\ell:\mathbb{R}^p \times \Theta \to \mathbb{R}$ be a  loss function, and 
define the population loss at time $t$ as
$F_t(\boldsymbol{\theta}) = \mathbb{E}\!\left[\ell(\boldsymbol{X}_t,\boldsymbol{\theta})\right].$ The target parameter at time $t$ is the minimizer of the population loss, i.e.,
$\boldsymbol{\theta}_t^*
= \arg\min_{\boldsymbol{\theta}\in\Theta} F_t(\boldsymbol{\theta}).$ In practice, the distribution $\mathbb{P}_t$ is unknown, so the population loss 
$F_t(\boldsymbol{\theta})$ cannot be obtained directly.  
However, since the data arrive sequentially, the past observations
$\{\bm x_1,\dots,\bm x_{t-1}\}$ are available for the forecast at time $t$. 
If the data were stationary, a natural estimate of $F_t(\boldsymbol{\theta})$
would be the full-sample average $\sum_{i=1}^{t-1}\ell(\bm x_i,\boldsymbol{\theta})/(t-1).$
However, in many cases, the distribution $\mathbb{P}_t$ may shift 
over time due to discrete or continuous structural changes. 
Including too much pre-break data in the estimation may introduce substantial 
estimation bias. 
A more robust strategy is therefore to approximate $F_t(\boldsymbol{\theta})$
using only a recent look-back window of length $k$, within which we believe that the distribution has no significant shift. 
This leads to the empirical loss
\begin{equation}\label{eq:emp_loss}
f_{t,k}(\boldsymbol{\theta})
= \frac{1}{k}\sum_{i=t-k}^{t-1} \ell(\bm x_i,\boldsymbol{\theta}),
\end{equation}
based on $\{\bm x_i\}_{t-k}^{t-1}$ and the corresponding minimizer $
\hat{\boldsymbol{\theta}}_{t,k}
\in \arg\min_{\boldsymbol{\theta}\in\Theta} f_{t,k}(\boldsymbol{\theta})$
serves as an approximation to the target parameter $\boldsymbol{\theta}_t^*$.

This paper focuses on selecting the largest window $\hat k_t$ in which the distribution shift remains negligible, and then obtaining the estimated parameter $\hat{\bm\theta}_{t,\hat k_t}$.
We follow the selection framework in \cite{huang2025stability}, which is based on a stability principle: A statistically more stable solution is preferred unless it is significantly worse. 
In other words,  if $\{\mathbb{P}_i\}_{i=t-k}^{t-1}$ are close, incorporating sufficient historical data can improve statistical efficiency  without introducing a significantly higher bias.

Following \cite{huang2025stability}, we construct a pairwise test
 \begin{equation}T_{i,k}=\begin{cases}
        1, ~\text{if}~f_{t,i}(\hat{\bm \theta}_{t,k})-f_{t,i}(\hat{\bm \theta}_{t,i})>\tau(t,i)\\
        0, ~\text{if}~f_{t,i}(\hat{\bm \theta}_{t,k})-f_{t,i}(\hat{\bm \theta}_{t,i})\leq\tau(t,i)
    \end{cases}\label{eq:test}
\end{equation}
for a candidate window size $k\in [t-1]$ with  $[t-1]=\{1,\dots,t-1\}$ and each reference window size $i< k$. We denote the candidate window set by $\mathcal{K}_t$. Since $\hat{\bm \theta}_{t,i}$ is the minimizer of $f_{t,i}(\bm\theta)$, $\hat{\bm \theta}_{t,k}$ cannot achieve a lower loss within the smaller window $i$, that is, $f_{t,i}(\hat{\bm \theta}_{t,k})- f_{t,i}(\hat{\bm \theta}_{t,i})\geq 0$.   $T_{i,k}=0$ would imply that $\hat{\bm\theta}_{t,k}$ is not significantly worse than $\hat{\bm\theta}_{t,i}$. If $T_{i,k}=0$ for all $i<k$, we say that the window $k$ is
admissible whenever $T_k=0$, meaning that no significant distributional shift 
is detected within $\{t-k,\dots,t-1\}$. 
Otherwise, we set $T_k=1$, indicating a potential distribution shift over this period.

Shorter windows incorporate most recent observations and are more likely to reflect the current distribution $\mathbb{P}_t$. Therefore, comparing each candidate window $k$ against all shorter windows provides a useful mechanism for detecting distributional changes over time. Our objective is to determine the largest admissible window $\hat k_t=\max\{k\in[t-1]:T_k=0\}$ and  set $\hat{\bm \theta}_{t}=\hat{\bm \theta}_{t,\hat k_t}$.  \subsection{Bootstrap-based threshold specification}\label{sec:boot}
While large windows help reduce estimation variance, they increase the risk of bias due to potential distributional shifts. Thus, a well-specified threshold function is critical for balancing the bias-variance  trade-off in estimation.

If no significant distribution shift occurs within a window $i$, the distribution $\{\mathbb{P}_{l}\}_{l=t-i}^{t-1}$ should, if not identical, be sufficiently similar.  It motivates us to develop a bootstrap-based method to construct the threshold $\tau(t,i)$. The procedure is as follows: First, choose a parameter $\beta\in (0,1)$, typically close to $1$, such as $0.9$. Then, for time $t$ and each reference window $i$,
 \begin{itemize}
     \item {\it Step 1.} Draw a sample of size $i$ with  replacement from observations $\{\bm x_{t-i},\dots, \bm x_{t-1}\}$ and denote the sample as $\{\bm x_{t-i}^{(b)},\dots,\bm x_{t-1}^{(b)}\}$.
     \item {\it Step 2.} Compute the bootstrapped objective function $f_{t,i}^{(b)}(\bm\theta)=\frac 1i\sum_{l=t-i}^{t-1}\ell(\bm x_l^{(b)},\bm\theta),$ and obtain the corresponding estimator by
     solving $\hat{\bm\theta}^{(b)}_{t,i}=\argmin_{\bm\theta\in \Theta} f_{t,i}^{(b)}(\bm\theta).$
    \item {\it Step 3.} Repeat Steps 1 and 2  $B$ times and derive the set $\{\hat{\bm\theta}^{(b)}_{t,i}\}_{b=1}^B$.
     \item {\it Step 4.} Calculate the empirical \(\beta\)-quantile of
 $\{f_{t,i}(\hat{\bm\theta}_{t,i}^{(b)})-f_{t,i}(\hat{\bm\theta}_{t,i}) \}$,
defined as the \(\lceil \beta B\rceil\)-th order statistic, and use it as the threshold \(\tau(t,i)\).
 \end{itemize}
 The hyperparameter $\beta$ can be interpreted as the asymptotic confidence level for an idealized pairwise comparison against window $k$ under the null; see Section \ref{subsec:iid_bootstrap}.
  When the data within a window are approximately independent and identically distributed, the empirical bootstrap method is useful and the above bootstrap procedure applies. For dependent data, a moving block bootstrap procedure is more suitable, as it accounts for temporal dependence; see \cite{kunsch1989jackknife}. In this case, we replace Step 1 by Step 1* and perform Steps 2-4 using the block sample obtained in Step 1*. For given $t$ and $i$, we define the block $\mathcal{B}_{t,i,l}=\{\bm x_{t-l-l_i+1},\dots,\bm x_{t-l}\}$ with the block length $l_i$ satisfying $l_i\to\infty$ and $m_i=\lfloor i/l_i\rfloor\to\infty$ as $i\to\infty$. We select $l_i=c\lceil i^{1/3}\rceil$ for some positive constant $c$.
 \begin{itemize}
     \item {\it Step~1$^{*}$}. Resample $m_i$ blocks with replacement from $\{\mathcal{B}_{t,i,1}\dots \mathcal{B}_{t,i,i-l_i+1}\}$ and arrange all elements of $m_i$ blocks in a sequence to get the bootstrapped sample $\{\bm x_{t-m_il_i}^{(b)},\dots,\bm x_{t-1}^{(b)}\}$.
 \end{itemize}
 \iffalse
Under the pairwise null hypothesis
\[
H^{t,i,k}_0:\ \{X_{t-k},\ldots,X_{t-i-1}\}\ \text{and}\ \{X_{t-i},\ldots,X_{t-1}\}\ \text{are identically distributed},
\]
define the \emph{theoretical} $\beta$-quantile
\[
q_{t,i,k}(\beta):=\inf\Bigl\{q\in\mathbb{R}:\ 
\mathbb{P}\bigl(f_{t,i}\!\big(\hat\theta_{t,k}\big)\;-\;f_{t,i}\!\big(\hat\theta_{t,i}\big)\le q\mid H^{t,i,k}_0\bigr)\ge \beta\Bigr\}.
\]
We approximate $q_{t,i,k}(\beta)$ by the bootstrap quantile $\hat q_{t,i}(\beta)$ computed from
\[
\Bigl\{\, f_{t,i}(\hat\theta^{(b)}_{t,i})-f_{t,i}(\hat\theta_{t,i})\,\Bigr\}_{b=1}^B,
\]
and set $\tau(t,i)=\hat q_{t,i}(\beta)$. 
 \fi
 
In principle, all window sizes in $[t-1]$ could be considered. 
However, too short windows tend to produce highly unstable estimates, so it is 
reasonable to set a minimum window length $k_0$.   For computational efficiency, we adopt an increasing-interval strategy to construct a sparse but representative set of candidate windows. Specifically, window lengths grow more coarsely as they become larger. For example, we may use increments of 5 for windows below 50, increments of 10 
for windows between 50 and 100, increments of 20 for windows between 100 and 300, 
and increments of 50 between 300 and 1000. 
Beyond length 1000, only increments of 100
are considered.   The candidate set is further dynamically adjusted. When selecting the window at time $t$, we incorporate the previously selected window $\hat k_{t-1}$ into the window set as a reference. Windows with   length smaller than $\hat{k}_{t-1}+1$ follow the increment rules above, while 
larger windows are expanded at increments of 50 starting from 
$\hat{k}_{t-1}+1$ up to $t-1$. 

Since the data arrive sequentially, the proposed procedure can be applied 
online to generate forecasts  as new observations become 
available. We refer to this framework as bootstrap-based adaptive window selection (BAWS), which is summarized in Algorithm \ref{alg1}. 
\section{Theoretical guarantee for the bootstrap threshold}\label{sec:theory}
This section provides theoretical support for the BAWS procedure introduced in Section 2. 
We first justify the bootstrap-based threshold under an idealized setting with no distributional shift, showing that it approximates the quantile of the population excess loss. 
Then, we study a single-break example, which explains why BAWS retains  large windows with no statistically significant shift and rejects overlong windows crossing a sufficiently strong structural break. 
\subsection{An asymptotic justification }
\label{subsec:iid_bootstrap}
To clarify the mechanism of the bootstrap threshold, this subsection provides an asymptotic explanation under an M-estimation setting.
We consider an asymptotic regime  in which \(t\to\infty\) and the window size \(n=n_t\to\infty\), with \(n_t\le t\). Throughout this subsection, we assume that
\(\bm X_{t-n},\ldots,\bm X_{t-1}\) are i.i.d.~from a common distribution and the threshold is derived via the empirical bootstrap procedures. The dependent-data extension via the moving block bootstrap is treated separately in Section~\ref{sec:mbb} of the Supplementary Material. Under this setting, we show that  $n(f_{t,n}(\hat{\bm\theta}^{(b)}_{t,n})-f_{t,n}(\hat{\bm\theta}_{t,n}))$ exhibits asymptotic behavior similar to that of  $n(F_{t}(\hat{\bm\theta}_{t,n})-F_{t}({\bm\theta}_{t}^*))$ as the window expands.

For notational simplicity,  this subsection writes
\[\bm\theta^*:=\bm\theta^*_t,~\hat{\bm\theta}_n:=\hat{\bm\theta}_{t,n},~\hat{\bm\theta}^{(b)}_n:=\hat{\bm\theta}^{(b)}_{t,n}\]
\[F({\bm\theta}):=F_{t}({\bm\theta}),~f_n(\bm\theta):=f_{t,n}(\bm\theta),~f_n^{(b)}(\bm\theta):=f_{t,n}^{(b)}({\bm\theta}).\]

Throughout, for bootstrap quantities, probabilistic statements are understood
conditionally on the observed sample and in probability. For a bootstrap
statistic \(\bm R_n^{(b)}\),  \(\bm R_n^{(b)}=o_p(1)\) means that
\(
\mathbb P^*(\|\bm R_n^{(b)}\|>\varepsilon)\xrightarrow{p}0
\)
for every \(\varepsilon>0\), where \(\mathbb P^*\) denotes probability conditional on the
observed sample. We denote by \(\mathcal L^*(\bm R_n^{(b)})\) the conditional distribution
of \(\bm R_n^{(b)}\) given the observed sample, and write
\(
\mathcal L^*(\bm R_n^{(b)})\Rightarrow_p\mathcal L(\bm R)
\)
for weak convergence of the conditional distribution in probability. Equivalently, for every
bounded Lipschitz function \(\varphi\),
\(
\mathbb E^*\!\left[\varphi(\bm R_n^{(b)})\right]
\xrightarrow{p}
\mathbb E\!\left[\varphi(\bm R)\right],
\)
where \(\mathbb E^*\) denotes expectation conditional on the observed sample.

We assume the following regularity conditions.
 \begin{itemize}
     \item {\it (C1) Identification.}  The parameter space \(\Theta\subseteq\mathbb R^d\) is compact, and the population loss
\(
F(\bm\theta)=\mathbb E[\ell(\bm X_t,\bm\theta)]
\)
admits a unique minimizer \(\bm\theta^*\in\operatorname{int}(\Theta)\).
     \item {\it (C2) Uniform consistency.} As \(n\to\infty\),
\(\sup_{\bm\theta\in\Theta}|f_n(\bm\theta)-F(\bm\theta)|
\xrightarrow{p}0,\)
and \(
\sup_{\bm\theta\in\Theta}|f_n^{(b)}(\bm\theta)-f_n(\bm\theta)|
\xrightarrow{p}0\),
where the second convergence is conditional on the observed sample, as defined above.
\item {\it (C3) Population smoothness.} The population loss $F$ is twice continuously differentiable in a neighborhood $\mathcal N(\bm\theta^*)$ of $\bm\theta^*$, with $\nabla F(\bm\theta^*)=0$ and Hessian $\Sigma:=\nabla^2 F(\bm\theta^*)$ positive definite.
\iffalse
Moreover,
\begin{equation}\label{eq:qmd}
F(\bm\theta) - F(\bm\theta^*) - \tfrac{1}{2}(\bm\theta-\bm\theta^*)^\top \Sigma (\bm\theta-\bm\theta^*) = o(\|\bm\theta-\bm\theta^*\|^2).
\end{equation}
    \item{\it (C4) Linear stochastic expansion.}
    There exists a measurable function $\bm\psi:\R^p\times\Theta\to\R^d$, called the generalized score, with $\E[\bm\psi(\bm X_t,\bm\theta^*)]=0$ and $\Omega:=\E[\bm\psi(\bm X_t,\bm\theta^*)\bm\psi(\bm X_t,\bm\theta^*)^\top]$ finite, such that
\begin{equation}\label{eq:linexp}
\sqrt{n}(\hat{\bm\theta}_n-\bm\theta^*) = -\Sigma^{-1}\, \frac{1}{\sqrt{n}}\sum_{l=1}^n \bm\psi(\bm X_{t-l},\bm\theta^*) + o_p(1),
\end{equation}
and the same expansion holds for the bootstrap estimator $\hat{\bm\theta}^{(b)}_n$ around $\hat{\bm\theta}_n$ conditionally on the observed sample, in probability.
\fi
\item {\it (C4) Boundedness.}
The empirical and  bootstrap estimators are stochastically bounded at the
root-\(n\) rate:
\[
\sqrt n\|\hat{\bm\theta}_n-\bm\theta^*\|=O_p(1),
\qquad
\sqrt n\|\hat{\bm\theta}_n^{(b)}-\hat{\bm\theta}_n\|=O_p(1).
\]
%where the second bound is  conditionally on the observed sample, in probability.
\item {\it (C5) Quadratic stochastic expansion.}
There exists a measurable function
\(\bm\psi:\mathbb R^p\times\Theta\to\mathbb R^d\), called the generalized
score, such that
\[
\E[\bm\psi(\bm X_t,\bm\theta^*)]=\bm 0,\qquad
\Omega:=\E[\bm\psi(\bm X_t,\bm\theta^*)\bm\psi(\bm X_t,\bm\theta^*)^\top]<\infty .
\]
Let
\(\bm Z_n=n^{-1/2}\sum_{l=1}^n\bm\psi(\bm X_{t-l},\bm\theta^*),~
\bar{\bm\psi}_n=n^{-1}\sum_{j=1}^n\bm\psi(\bm X_{t-j},\bm\theta^*),
\)
and, for the bootstrap sample,
\(\bm Z_n^{(b)}
=
n^{-1/2}\sum_{l=1}^n
(
\bm\psi(\bm X_{t-l}^{(b)},\bm\theta^*)-\bar{\bm\psi}_n
).
\)
For every \(M>0\),
\begin{equation}\label{eq:second_expan1}
\sup_{\|\bm h\|\le M}
\left|
n(f_n(\bm\theta^*+\bm h/\sqrt n)-f_n(\bm\theta^*))
-\bm h^\top\bm Z_n
-\tfrac12\bm h^\top\Sigma\bm h
\right|
\xrightarrow{p}0,
\end{equation}
and
\begin{equation}\label{eq:second_expan2}
\sup_{\|\bm h\|\le M}
\left|
n(f_n^{(b)}(\hat{\bm\theta}_n+\bm h/\sqrt n)-f_n^{(b)}(\hat{\bm\theta}_n))
-\bm h^\top\bm Z_n^{(b)}
-\tfrac12\bm h^\top\Sigma\bm h
\right|
\xrightarrow{p}0
\end{equation}
conditionally on the observed sample, in probability.
\end{itemize}
\begin{remark}\label{rem:revised_conditions}
Conditions (C1)--(C5) impose smoothness only on the population object $F$ (Condition (C3)) rather than on the loss
\(\ell(\bm x,\cdot)\). Condition (C4) ensures that the empirical and bootstrap
estimators are  bounded in probability, so that the local expansion
in Condition (C5) can be applied to the random minimizers. These ingredients are exactly what is used in the proof of Theorem~\ref{thm:limit}, and they accommodate non-smooth losses such as the check function and the Fissler-Ziegel score, as we verify in Section~\ref{app:verify} of the Supplementary Material. In particular, Conditions (C1)--(C5) contain the usual smooth M-estimation as a special case, like mean estimation under the squared loss. When $\ell(\bm x,\cdot)$ is twice continuously differentiable and \(\nabla_{\bm\theta}^2\ell\) admits an integrable Lipschitz envelope, the expansion in (C5) follows from the Taylor expansion. 
\end{remark}
\begin{theorem}\label{thm:limit}
Under Conditions (C1)--(C5),
\begin{equation}\label{eq:limit_dist}
\sqrt{n}(\hat{\bm\theta}_n-\bm\theta^*)\overset{d}{\to}\bm Z
~\text{and}~n\left(F(\hat{\bm\theta}_n)-F(\bm\theta^*)\right)\overset{d}{\to}\frac{1}{2} \bm Z^\top\Sigma \bm Z,
\end{equation}
where \(\bm Z\sim N(0,\tilde{\Sigma})\) and \(\tilde{\Sigma}=\Sigma^{-1}\Omega\Sigma^{-1}\).  Moreover,
\[
\mathcal L^*\!\left(
n\left(f_n(\hat{\bm\theta}^{(b)}_n)-f_n(\hat{\bm\theta}_n)\right)
\right)
\Rightarrow_p
\mathcal L\!\left(
\frac12 \bm Z^\top\Sigma\bm Z
\right).
\]     That is, the conditional distribution of the bootstrap excess loss converges in probability
to the same quadratic-form limit.
\end{theorem}
Theorem \ref{thm:limit} implies that the conditional limiting distribution of $n(f_n(\hat{\bm\theta}^{(b)}_n)-f_n(\hat{\bm\theta}_n))$ is a quadratic form of a Gaussian random vector when the data within the window are i.i.d.~and the size goes to $\infty$. In the special case \(\Omega=\Sigma\), the limiting quadratic form reduces to
\(\chi_d^2/2\). Therefore, an empirical $\beta$-quantile of bootstrap excess losses is an asymptotic approximation of the $\beta$-quantile of $F(\hat{\bm\theta}_n)-F(\bm\theta^*)$ for sufficiently large $n$.
\begin{remark}
When the observations within a window are
stationary and weakly dependent, a similar argument can be combined with the moving
block bootstrap by replacing \(\Omega\) with the long-run variance
\[
\Omega_{\mathrm{LR}}
:=
\sum_{j\in\mathbb Z}
\E\!\left[
\bm\psi(\bm X_t,\bm\theta^*)
\bm\psi(\bm X_{t+j},\bm\theta^*)^\top
\right]
\]
which does not depend on \(t\) by  stationarity. Under standard mixing and block-length conditions, the bootstrap excess loss then converges
to the corresponding quadratic Gaussian limit with covariance
\(\Sigma^{-1}\Omega_{\mathrm{LR}}\Sigma^{-1}\). A formal statement is provided in
Section \ref{sec:mbb} of the Supplementary Material. This extension justifies the moving block bootstrap 
used for dependent data in Sections~5.3 and~6.
\end{remark}

To further clarify the role of the bootstrap threshold in the pairwise test \eqref{eq:test}, we consider an asymptotic argument under the null hypothesis
\[H_0^{t,k}:~\text{no distributional shift occurs within a window}~ k.\]
Let $k>n$, so that both \(\hat{\bm\theta}_n\) and \(\hat{\bm\theta}_k\) target the same population minimizer \(\bm\theta^*\).

Suppose that \(n,k\to\infty\) and \(k/n\to\infty\). Under \(H_0^{t,k}\), both
\(\hat{\bm\theta}_n\) and \(\hat{\bm\theta}_k\) target the same population minimizer
\(\bm\theta^*\). Applying Condition~(C4) to the larger window \(k\), we have
\(
\sqrt{k}\|\hat{\bm\theta}_k-\bm\theta^*\|=O_p(1),
\)
and hence
\(
\hat{\bm\theta}_k-\bm\theta^*=o_p(n^{-1/2}).
\)
Therefore,
\(
\sqrt n(\hat{\bm\theta}_k-\hat{\bm\theta}_n)
=
-\sqrt n(\hat{\bm\theta}_n-\bm\theta^*)+o_p(1).
\)

Let
\(
\bm h_1=\sqrt n(\hat{\bm\theta}_n-\bm\theta^*),\) \(
\bm h_2=\sqrt n(\hat{\bm\theta}_k-\bm\theta^*).
\)
By Condition~(C4), \(\bm h_1=O_p(1)\), and the previous result gives
\(\bm h_2=o_p(1)\). Applying the quadratic
expansion \eqref{eq:second_expan1} in Condition~(C5) at \(\bm h=\bm h_2\) and \(\bm h=\bm h_1\), and subtracting
the two expansions, we obtain
\[
\begin{aligned}
n\bigl(f_n(\hat{\bm\theta}_k)-f_n(\hat{\bm\theta}_n)\bigr)
&=
\tfrac12\bm h_2^\top\Sigma\bm h_2+\bm h_2^\top \bm Z_n
-\tfrac12\bm h_1^\top\Sigma\bm h_1-\bm h_1^\top \bm Z_n
+o_p(1).
\end{aligned}
\]
By  Conditions~(C4)--(C5) through the argmin continuous mapping theorem \citep[Theorem~3.2.2]{vandervaart1996weak}, we obtain that
\(
\bm h_1=-\Sigma^{-1}\bm Z_n+o_p(1).
\)
Since \(\bm Z_n=O_p(1)\) and \(\bm h_2=o_p(1)\), the terms involving \(\bm h_2\) are
\(o_p(1)\). Hence,
\[
n\bigl(f_n(\hat{\bm\theta}_k)-f_n(\hat{\bm\theta}_n)\bigr)
=
\frac12
\left(\sqrt n(\hat{\bm\theta}_n-\bm\theta^*)\right)^\top
\Sigma
\left(\sqrt n(\hat{\bm\theta}_n-\bm\theta^*)\right)
+o_p(1).
\]
On the other hand, by Taylor's expansion and Condition~(C3),
\[
n\bigl(F(\hat{\bm\theta}_n)-F(\bm\theta^*)\bigr)
=
\frac12
\left(\sqrt n(\hat{\bm\theta}_n-\bm\theta^*)\right)^\top
\Sigma
\left(\sqrt n(\hat{\bm\theta}_n-\bm\theta^*)\right)
+o_p(1).
\]
Therefore,
\(
n\bigl(f_n(\hat{\bm\theta}_k)-f_n(\hat{\bm\theta}_n)\bigr)
\)
shares the same asymptotic limit as
\(
n\bigl(F(\hat{\bm\theta}_n)-F(\bm\theta^*)\bigr)
\)
under \(H_0^{t,k}\). 
Combining this observation with Theorem~\ref{thm:limit}, the bootstrap threshold
\(\tau(t,n)\) serves as an asymptotic critical value  for the
pairwise comparison. Let
\(R=\frac12\bm Z^\top\Sigma\bm Z,\) \(q_\beta=\inf\{x:\mathbb P(R\le x)\ge \beta\}.
\)
Suppose that the distribution of \(R\) is
continuous at \(q_\beta\) and that the bootstrap threshold satisfies
\(n\tau(t,n)\xrightarrow{p}q_\beta.\)
Then, under \(H_0^{t,k}\) and \(k/n\to\infty\),
\begin{equation}
    \label{eq:beta} 
\mathbb P\left(
f_n(\hat{\bm\theta}_k)-f_n(\hat{\bm\theta}_n)>\tau(t,n)
\,\middle|\,H_0^{t,k}
\right)
\to 1-\beta.
\end{equation} 
Thus, \(1-\beta\) can be interpreted as the asymptotic type-I error level of an individual
pairwise comparison.

The above discussion concerns a single pairwise comparison. The full BAWS procedure compares a candidate window $k$ with a collection of overlapping reference windows. These windows are dependent in a complicated way, and we do not have an interpretation of $\beta$ as in \eqref{eq:beta} for the full procedure. 
To obtain  a conservative benchmark for the accumulation of type-I error,  one may use a
Bonferroni correction. In the next result, let
\(
s_{t,k}:=\left|\{i\in\mathcal K_t:i<k\}\right|.
\)
\begin{proposition}[Bonferroni]
\label{prop:bonfe}
For a significance level \(\alpha\in(0,1)\), suppose that the pairwise thresholds \(\tau_{\mathrm{Bon}}(t,i)\) satisfy
$\mathbb P\left(
f_i(\hat{\bm\theta}_k)-f_i(\hat{\bm\theta}_i)
>
\tau_{\mathrm{Bon}}(t,i)
\,\middle|\, H_{0}^{t,k}
\right)
\le
\alpha/s_{t,k},$
for $i\in\mathcal K_t$ and $i<k$.
Then
$
\mathbb P\left(
T_k=1
\,\middle|\, H_0^{t,k}
\right)
\le
\alpha.$
 \end{proposition}
\begin{proof}
The result directly follows from
\(
\{T_k=1\}
=
\bigcup_{i\in\mathcal K_t:\,i<k}
\{
f_i(\hat{\bm\theta}_k)-f_i(\hat{\bm\theta}_i)
>
\tau_{\mathrm{Bon}}(t,i)
\}.
\)
\end{proof}
In practice, \(\tau_{\mathrm{Bon}}(t,i)\) can be computed using the same bootstrap procedure as \(\tau(t,i)\), with the threshold level \(
1-\alpha/s_{t,k}\)
for the pairwise comparison between \(k\) and \(i\). 
Since pairwise tests are nested and strongly dependent,  the Bonferroni correction
is generally conservative. We therefore use the global tuning parameter \(\beta\) in the
numerical studies. 
\subsection{An illustration:  single structural break}\label{sec:example}
A full characterization of the BAWS procedure is challenging in general nonstationary settings, since the window selection depends on a data-driven bootstrap threshold and a nested sequence of pairwise comparisons.
To obtain clearer analytical insight, we consider a simple single-break setting in which the performance of the pairwise test can be described explicitly. 

%For the proof we only use the squared-loss specialization, under which the empirical risk minimizer is $\hat\mu_{t,k}=\bar X_{t,k}$
Given a sequence of random variables \(\{X_t\}\), we work with mean estimation under the squared loss 
\(\ell(x,\mu)=(x-\mu)^2\). 
For a window of size \(k\), the empirical minimizer is the sample mean
$\hat\mu_{t,k}
=
\bar X_{t,k}
:=
\frac{1}{k}\sum_{l=t-k}^{t-1}X_l .$
The corresponding statistic has the exact form
$$f_{t,i}(\hat\mu_{t,k})-f_{t,i}(\hat\mu_{t,i}) =
\bigl(\bar X_{t,k}-\bar X_{t,i}\bigr)^2,
\qquad i<k.$$

Under the homogeneous null $H_0^{t,k}$, all observations within 
\(\{t-k,\ldots,t-1\}\) are generated from the same distribution, so the larger-window estimator is expected to remain compatible with the recent reference window. 
We now consider the following single-break alternative:
\[H_1:
X_{1},\ldots,X_{t-k_0-1}\overset{\text{i.i.d}}{\sim} \mathbb P_1,
\qquad
X_{t-k_0},\ldots,X_{t-1}\overset{\text{i.i.d}}{\sim} \mathbb P_2,
\]
with $\mathbb{E}_{\mathbb P_1}[X]=\mu_1$ and $\mathbb{E}_{\mathbb P_2}[X]=\mu_2\neq \mu_1$. 
Assume further that $\operatorname{Var}_{\mathbb P_1}(X)<\infty$ and $\operatorname{Var}_{\mathbb P_2}(X)<\infty$, where  $\mathbb E_{\mathbb P_i}[X]$ and $\operatorname{Var}_{\mathbb P_i}(X)$ denote the mean and variance under $\mathbb P_i$ for $i\in\{1,2\}$, respectively. Therefore, for a window $k>k_0$,  $\{t-k,\ldots,t-1\}$ contains a single change point at $t-k_0$. Under this alternative, the larger window \(k\) mixes observations from two regimes, 
whereas the reference window \(k_0\) contains only post-break observations. 

Assume that \(k_0\to\infty\) as \(t\to\infty\). 
Let \(k=k_t\in\mathcal K_t\) be a candidate window satisfying \(k>k_0\) and
\(
k_0/k\to 1-c
\)
for some constant \(c\in(0,1]\).
As in Section~\ref{sec:problem}, \(T_k\) denotes the window-level decision at
time \(t\): \(T_k=0\) means that window \(k\) is admissible, while \(T_k=1\)
means that it is rejected.

The following result formalizes the rejection of the window $k$ under $H_1$.
\begin{theorem} 
\label{thm:reject_large_window}
 For some $\epsilon$ satisfying $0<\epsilon<c^2(\mu_1-\mu_2)^2$, if
\begin{equation}\label{con:limit_tau}
\mathbb P(\tau(t,k_0)\le c^2(\mu_1-\mu_2)^2-\epsilon\mid H_1)\to 1~\text{as}~t\to\infty,
\end{equation} then as $t\to\infty$,
\[
\mathbb{P}(T_{k}=0\mid H_1)\to 0,
\qquad\text{equivalently}\qquad
\mathbb{P}(T_{k}=1\mid H_1)\to 1.
\]
\end{theorem}

Condition \eqref{con:limit_tau} requires the bootstrap threshold to be asymptotically smaller than the mean-shift signal \(c^2(\mu_1-\mu_2)^2\). 
This is natural in view of Theorem~\ref{thm:limit}: for the squared loss, under Conditions (C1)-(C5), if no distributional shift occurs within the reference window  \(k_0\), then the bootstrap threshold \(\tau(t,k_0)\) is asymptotically approximated by \(q_\beta/k_0\), where \(q_\beta\) is the \(\beta\)-quantile of the limiting quadratic form in \eqref{eq:limit_dist}. 
Hence, the threshold vanishes as \(k_0\to\infty\), while the break signal \(c^2(\mu_1-\mu_2)^2\) is positive. It follows that \eqref{con:limit_tau} holds with probability tending to one under \(H_1\), and BAWS rejects the overlong window with probability tending to one. Furthermore, we obtain the following result for the full procedure.
\begin{corollary}
\label{cor:selected_window}
Let \(\mathcal K_t^+=\{l\in\mathcal K_t:l\ge k\}\) and \(r_t=|\mathcal K_t^+|\). Assume that \(k_0\in\mathcal K_t\) and \(k_0<k\). Suppose that Condition \eqref{con:limit_tau} holds and that
\begin{equation}\label{con:sparse}
\frac{r_t}{k-k_0}\to0~\text{as}~t\to\infty.
\end{equation}
\iffalse
Assume further that
\begin{equation}\label{con:reject_large_window}
\mathbb P\left(
\inf_{l\in\mathcal K_t:\,l\ge k}
\left\{
f_{t,k_0}(\hat\mu_{t,l})-f_{t,k_0}(\hat\mu_{t,k_0})
\right\}
>
\tau(t,k_0)
\,\middle|\,H_1
\right)\to1.
\end{equation}
\fi
Then as $t\to\infty$,
$\mathbb P(\hat k_t\ge k\mid H_1)\to 0.$
\end{corollary}

 Condition \eqref{con:sparse} reflects the sparse construction of candidate window sets in BAWS described in Section \ref{sec:boot}.
Since the increment of candidate windows becomes larger for longer windows, the number of windows in \(\mathcal K_t^+\) grows  more slowly than \(k-k_0\).   Under the single-break hypothesis \(H_1\), this condition ensures that 
\[
\mathbb P\left(
\bigcap_{l\in\mathcal K_t^+}\{T_l=1\}
\,\middle|\,H_1
\right)\to1.
\]
Thus,  \(\hat k_t<k\) with probability tending to one. 

\begin{example}[Gaussian special case]
To obtain an explicit expression for the rejection probability, suppose that the observations in the two regimes are independent Gaussian:
\[
X_{t-k},\ldots,X_{t-k_0-1}\stackrel{\text{i.i.d.}}{\sim} N(\mu_1,\sigma_1^2),
\qquad
X_{t-k_0},\ldots,X_{t-1}\stackrel{\text{i.i.d.}}{\sim} N(\mu_2,\sigma_2^2),
\]
with the two blocks independent and \(\mu_1\ne\mu_2\). Define
\[
\mu_{t,k_0}
:=
\frac{k-k_0}{k}(\mu_1-\mu_2),
\quad
\sigma^2_{t,k_0}
:=
\left(\frac{k-k_0}{k}\right)^2
\left(
\frac{\sigma_1^2}{k-k_0}
+
\frac{\sigma_2^2}{k_0}
\right).
\]
Then
$\bar X_{t,k}-\bar X_{t,k_0}
\sim
N(\mu_{t,k_0},\sigma^2_{t,k_0}).$
Therefore, for any fixed deterministic threshold \(\tau(t,k_0)>0\), under \(H_1\),
\[
\begin{aligned}
\mathbb P(T_k=1\mid H_1)
&\ge
\mathbb P\left(
(\bar X_{t,k}-\bar X_{t,k_0})^2>\tau(t,k_0)
\,\middle|\,H_1
\right) \\
&=
1-\Phi\left(
\frac{\sqrt{\tau(t,k_0)}-\mu_{t,k_0}}{\sigma_{t,k_0}}
\right)
+
\Phi\left(
\frac{-\sqrt{\tau(t,k_0)}-\mu_{t,k_0}}{\sigma_{t,k_0}}
\right).
\end{aligned}
\]
where \(\Phi\) denotes the standard normal cumulative distribution function.
\end{example}

\begin{example}[Two windows comparison]
\label{ex:naive_select_250}
Consider a single break at time \(250\). Specifically,
\(X_1,\ldots,X_{250}\) are i.i.d.~from \(\mathbb P_1\), whereas
\(X_{251},\ldots,X_{500}\) are i.i.d.~from \(\mathbb P_2\). Let
\(\mathbb E_{\mathbb P_1}(X)=\mu_1\) and
\(\mathbb E_{\mathbb P_2}(X)=\mu_2\ne\mu_1\), and assume that both
distributions have finite variances. We estimate the mean at time
\(t=501\) under squared loss and restrict the candidate windows to
\(\{250,500\}\). The window of length \(250\) contains only post-break
observations, whereas the window of length \(500\) mixes the two regimes.

Define
\[
\bar X_1=\frac{1}{250}\sum_{l=1}^{250}X_l,
\qquad
\bar X_2=\frac{1}{250}\sum_{l=251}^{500}X_l .
\]
Then
\(
\hat\mu_{t,250}=\bar X_2\) and \(
\hat\mu_{t,500}=(\bar X_1+\bar X_2)/2.
\)
Hence, by the squared-loss identity,
\[
f_{t,250}(\hat\mu_{t,500})
-
f_{t,250}(\hat\mu_{t,250})
=
(\hat\mu_{t,500}-\hat\mu_{t,250})^2
=
\frac{1}{4}(\bar X_1-\bar X_2)^2.
\]
For large blocks, this quantity is close to \((\mu_1-\mu_2)^2/4\). Therefore, if the break signal \((\mu_1-\mu_2)^2/4\) exceeds the threshold \(\tau(t,250)\), the window \(500\) fails the stability test against the reference window \(250\), and BAWS selects \(250\).
\end{example}

\section{Window selection for VaR and ES forecasts}\label{sec:risk_fore}

Elicitability is useful in model selection, forecast comparison, and backtesting of financial risk measures; see \cite{gneiting2011making} and \cite{fissler2015expected}.  
We now apply BAWS to forecast elicitable risk measures, specifically VaR and ES, which are important for financial regulation and portfolio management practice; for VaR and ES in regulation, see  \cite{embrechts2014academic} and \cite{mcneil2015quantitative}.

Let $X_t \sim \mathbb{P}_t$ denote the random variable of the financial loss at time $t$, and $\{x_1,\dots,x_{t-1}\}$ be the observed historical losses over time. Recall that the VaR at level $\alpha\in (0,1)$ for the loss $X_t$ is defined as
$\VaR_{\alpha}(X_t)=\inf\{x\in\R: \mathbb{P}(X_t\leq x)\geq \alpha\}.$
Since VaR is elicitable~\citep{gneiting2011making},
it can be characterized as the minimizer of an expected scoring function:
\begin{equation}\label{eq:VaR}
    \VaR_{\alpha}(X_t)=\argmin_{v\in\mathbb{R}} \mathbb{E}\left[S_{V,\alpha}(X_t,v)\right]
    \end{equation} 
   where
    \begin{equation}\label{eq:score_VaR}
        S_{V,\alpha}(x,v)=({\bf 1}\{x < v\}-\alpha)\left[G(v)-G(x)\right]
    \end{equation}
    is a scoring function, \(G\) is strictly increasing,
and \(\E[G(X_t)]\) exists; see \cite{fissler2015expected}.
 Specifically, we can take $G(x)=x$ and then $S_{V,\alpha}(x,v)$ is known as the check function. If \(X_t\) has a unique \(\alpha\)-quantile, then
\(\VaR_\alpha(X_t)\) is the unique minimizer of \eqref{eq:VaR}. %Under the assumption of continuity of the quantile function of $X_t$ at $\alpha$, $\VaR_\alpha(X_t) $ is  the unique minimizer of (\ref{eq:VaR}). 

For the random variable $X_t$ with a finite mean, Expected Shortfall (ES) is defined as 
\begin{equation*}
\ES_\alpha(X_t)= \frac{1}{1-\alpha}\int_\alpha^1 \VaR_\beta(X_t) \d\beta.
\end{equation*}
Since the pair $(\VaR_\alpha,\ES_\alpha)$ is jointly elicitable  \citep{fissler2016higher}, we can obtain the pair by the following optimization problem
\begin{equation}\label{eq:ES}
\left(\VaR_\alpha(X_t),\ES_\alpha(X_t)\right)=\argmin\limits_{(v,e)\in \R^2}\mathbb{E}\left[S_{V,E, \alpha}(X_t,v,e)\right],
\end{equation}
with
\begin{align}\label{eq:score_joint}
 S_{V,E, \alpha}(x,v,e)&=(\textbf{1}\{x< v\}-\alpha)\left[G_1(v)-G_1(x)\right]
 +\frac 1{1-\alpha}G_2(e)\textbf{1}\{x\geq v\}(v-x)\nonumber\\
 &\qquad+G_2(e)(e-v)-{\mathcal{G}}_2(e),  
\end{align}
where $G_1$ and $G_2$  are strictly increasing and continuously differentiable such that the expectation $\E \left[G_1(X_t)\right]$ exists,  and 
$\lim_{x\to +\infty}G_2(x)=0~\text{and}~{\mathcal{G}}^\prime_2=G_2.$

As we define the risk measure based on  the loss rather than the return, (\ref{eq:score_joint}) and the corresponding conditions are induced from Corollary 5.5 of \cite{fissler2016higher}. Note that any specification of $G_1(x)$ and $G_2(x)$ satisfying the above properties will lead to the unique minimizer of (\ref{eq:ES}). Here, we set
$G_1(x)=x$ and $G_2(x)=-\frac{\exp(-x)}{1+\exp(-x)}$, following \cite{fissler2016higher}.

When the distribution $\mathbb{P}_t$ is unknown, expected scores in (\ref{eq:VaR}) and (\ref{eq:ES}) need to be approximated using the sample average of the observed data. Given a window of size $k$ and observations $\{x_i\}_{i=t-k}^{t-1}$, we estimate VaR and ES via empirical loss functions, i.e.,
\begin{equation}\label{eq:VaR1}
    \widehat\VaR_{\alpha}(X_t)=\argmin_{v} \frac{1}{k}\sum_{l=t-k}^{t-1}S_{V, \alpha}(x_l,v)
\end{equation} 
and
\begin{equation}\label{eq:ES1}
\left(\widehat{\VaR}_\alpha(X_t),\widehat{\ES}_\alpha(X_t)\right)=\argmin\limits_{v,e}\frac{1}{k}\sum_{l=t-k}^{t-1}S_{V,E, \alpha}(x_l,v,e),
\end{equation}
 where $\widehat\VaR_\alpha(X_t)$ and $\widehat\ES_\alpha(X_t)$ denote empirical VaR and ES, respectively.

To cope with potential structural changes in financial losses, 
we apply the proposed window-selection framework to VaR and ES forecasting, 
using the empirical loss functions in \eqref{eq:VaR1} and \eqref{eq:ES1} 
as particular cases of $f_{t,k}$ in \eqref{eq:emp_loss}. 

Beyond VaR and ES, the elicitability of many other classes of risk measures have been characterized in the literature. For instance,  coherent risk measures have been studied by \cite{ziegel2016coherence}, convex risk measures by \cite{bellini2015elicitable} and \cite{delbaen2016risk}, tail risk measures by \cite{liu2021theory} and \cite{fissler2024elicitability}, and distortion risk measures by \cite{kou2016measurement} and \cite{wang2015elicitable}. Therefore, our proposed window selection strategy could be extended to other elicitable risk measures via the corresponding scoring functions.
In view of the importance of VaR and ES in financial practice,  we omit a discussion of scoring functions for other risk measures. 

Note that the  regularity conditions in Section~\ref{subsec:iid_bootstrap}   accommodate the non-smooth scoring functions used here. 
For the theoretical
verification, the empirical argmins in \eqref{eq:VaR1}--\eqref{eq:ES1} are
understood on the compact parameter spaces specified in Section~\ref{app:verify}
of the Supplementary Material. Under standard quantile-density and moment
conditions, the check function \eqref{eq:score_VaR} and the Fissler-Ziegel score
\eqref{eq:score_joint} satisfy Conditions~(C1)--(C5).
 \section{Simulation studies} \label{sec:Sim}
In this section, we conduct three simulation studies to evaluate the forecast performance of BAWS compared to various existing approaches. These simulations are designed to reflect real-world changes in financial markets, including abrupt breaks in market conditions, continuous changes in market trends, and time-varying volatility in market risk. Specifically, we focus on VaR forecasting in the main text, as it is the primary risk measure of interest in this paper. Mean forecasting is used as an auxiliary elicitable benchmark, and the corresponding numerical results are reported in the Supplementary Material.%The details of these simulations are provided in Sections \ref{Sec:Sim1}--\ref{Sec:Sim3}. 
 \subsection{Scenario 1: Discrete structural changes}\label{Sec:Sim1}
 This study investigates independent data with structural breaks over time, where a structural break refers to an abrupt change in the data-generating process, such as a sudden shift in the mean or variance parameter.

The simulated data $X_t$ is independently generated from the normal distribution $N(\mu_t,\sigma^2_t)$ with mean $\mu_t$ and variance $\sigma^2_t$ for $t=1,\dots, T$, where $T=2000$. For a comprehensive evaluation, we consider three parameter settings in this scenario.
\begin{itemize}
\item \textbf{Setting A1 (Two regimes with an abrupt mean shift).}  
The process has constant variance \(\sigma_t^2=0.25\). The mean changes from 
\(\mu_t=1\) for \(t\le T/2\) to \(\mu_t=2\) for \(t>T/2\).

\item \textbf{Setting A2 (Three regimes with piecewise-constant mean).}  
The process has constant variance \(\sigma_t^2=0.25\). The mean is 
\(\mu_t=1\) for \(t\le 800\), \(\mu_t=0\) for \(800<t\le 1400\), and 
\(\mu_t=2\) for \(t>1400\).

\item \textbf{Setting A3 (Structural breaks in both mean and variance).}  
The mean follows the same three-regime pattern as in Setting~A2. The variance is 
\(\sigma_t^2=0.25\) for \(t\le 800\), \(\sigma_t^2=1\) for \(800<t\le 1400\), 
and \(\sigma_t^2=0.49\) for \(t>1400\).
\end{itemize}
\iffalse
\begin{itemize}
\item \textbf{Setting A1 (Two regimes with an abrupt mean shift).}  
The process has constant variance and a single break in the mean:
\[
\mu_t = 
\begin{cases}
1, & t \le T/2,\\
2, & t > T/2,
\end{cases}
\qquad 
\sigma_t^2 = 0.25.
\]

\item \textbf{Setting A2 (Three regimes with piecewise-constant mean).}  
The variance remains constant, while the mean exhibits two structural breaks:

\[
\mu_t =
\begin{cases}
1, & t \le 800,\\
0, & 800 < t \le 1400,\\
2, & t > 1400,
\end{cases}
\qquad 
\sigma_t^2 = 0.25.
\]

\item \textbf{Setting A3 (Structural breaks in both mean and variance).}  
The mean follows the same three-regime pattern as in Setting~A2, 
while the variance changes as:
\[
\sigma_t^2 =
\begin{cases}
0.25, & t \le 800,\\
1, & 800 < t \le 1400,\\
0.49, & t > 1400.
\end{cases}
\]

\end{itemize}
\fi
The data  are therefore piecewise stationary and an appropriate window is required for forecasting. We set the threshold level to $\beta = 0.9$ and estimate thresholds using 
$500$ bootstrap replications. Each  experiment is replicated 
$n=1000$ times.
  The forecasting procedure for both the mean and VaR begins at $t_0=501$, and we compare our method with the SAWS  approach of \cite{huang2025stability},  the rolling window approach, and  full window approach.  We consider the risk level of $\alpha=0.95$ in VaR forecasting.

For the SAWS approach, we follow the parameter settings in Section~7.1 of 
\cite{huang2025stability}, motivated by the fact that the expected loss 
functions for mean and VaR forecasting are, respectively, strongly convex and smooth, 
and Lipschitz continuous. Specifically, we set $\alpha_\tau = 0.1$ and 
$C_{\tau} = 0.3$ for mean forecasts, and $\alpha_\tau = 0.1$ and 
$C_{\tau} = 0.5$ for VaR forecasts throughout Section \ref{sec:Sim}.
At each time $t$, the fixed-window approach uses rolling windows of 
sizes $k \in \{250, 500, 750\}$, corresponding approximately to one-, two-, 
and three-year windows. The full-window approach incorporates all available 
historical observations up to time $t-1$. For the fixed window with $k=750$, 
the full window is used whenever the past observations are fewer than 750.

\begin{table}[!ht]
\caption{The bias, variance, MSE, cumulative risk, and forecast loss for  VaR across BAWS, SAWS, fixed windows (250, 500, 750), and full window across Settings A1--A3.}
\scriptsize
\renewcommand{\arraystretch}{0.8}
\resizebox{\textwidth}{!}{
\begin{tabular}{clcccccc}
 \iffalse
\toprule
  \multirow{2}{*}{Mean} & & \multirow{2}{*}{BAWS} & \multirow{2}{*}{SAWS} &\multicolumn{3}{c}{Fixed Window} & \multirow{2}{*}{Full}\\
  \cline{5-7}
 &&&&250 & 500 & 750 \\

\midrule
 \multirow{4}{*}{A1} &MAB     & 0.0077 & 0.0083 & 0.0845 & 0.1678 & 0.2511 & 0.4631 \\
                       &Var & 0.0079 & 0.0073 & 0.0010 & 0.0005 & 0.0003 & 0.0002 \\
                       & CR       & 17.9276 & \textbf{14.8314} & 85.5601 & 168.4481 & 251.7759 & 501.8624 \\
                       & CL    & 392.9567   & \textbf{389.8877} & 460.3182 & 543.2237 & 626.4365 & 876.6553 \\
\cline{2-8}
  \multirow{4}{*}{A2} & MAB    & 0.0169 & 0.0123 & 0.2514 & 0.5010 & 0.7106 & 0.7744 \\
                       & Var & 0.0086 & 0.0088 & 0.0010 & 0.0005 & 0.0003 & 0.0002 \\
                       & CR       & 41.7347 & \textbf{22.6393} & 420.4758 & 836.3033 & 1188.3080 & 1201.3810 \\
                      &CL    & 417.3304& \textbf{397.9329} & 795.9474 & 1211.8626 & 1564.1570 & 1577.5416 \\
\cline{2-8}
  \multirow{4}{*}{A3} & MAB     & 0.0216 & 0.0133 & 0.2516 & 0.5010 & 0.7106 & 0.7744 \\
                       & Var & 0.0227 & 0.1941 & 0.0025 & 0.0012 & 0.0008 & 0.0004 \\
                      & CR       & \textbf{64.1937} & 297.7890 & 422.8770 & 837.5919 & 1189.3490 & 1201.7770 \\
                      & CL     & \textbf{1034.7530} & 1265.8910 & 1393.0080 & 1807.9140 & 2160.1410 & 2173.0750 \\
                       \fi
\toprule
 \multirow{2}{*}{VaR} & & \multirow{2}{*}{BAWS} & \multirow{2}{*}{SAWS} &\multicolumn{3}{c}{Fixed Window} & \multirow{2}{*}{Full}\\
  \cline{5-7}
 &&&&250 & 500 & 750 \\

\midrule
  \multirow{5}{*}{A1} & MAB    & 0.0597 & 0.1603 & 0.0478 & 0.0968 & 0.1360 & 0.2536 \\
                 & Var & 0.0055 & 0.0131 & 0.0047 & 0.0025 & 0.0018 & 0.0013 \\
                       & MSE      & \textbf{0.0191} & 0.0722 & 0.0282 & 0.0501 & 0.0710 & 0.1254 \\
                       & CR       & \textbf{5.0391} & 20.3860 & 7.8254 & 14.5752 & 20.7891 & 35.3949 \\
                       & CL     & \textbf{82.3192} & 97.7356 & 85.1195 & 91.9120 & 98.1440 & 112.7323 \\
\cline{2-8}
  \multirow{5}{*}{A2} & MAB    & 0.0853 & 0.4371 & 0.1807 & 0.3592 & 0.4864 & 0.6068 \\
                       & Var & 0.0082 & 0.0092 & 0.0052 & 0.0029 & 0.0021 & 0.0013 \\
                       & MSE      & \textbf{0.0587} & 0.3734 & 0.1563 & 0.3030 & 0.3927 & 0.5034 \\
                      & CR       & \textbf{11.9871} & 31.5068 & 26.3002 & 51.5869 & 54.4521 & 72.5658 \\
                       & CL     & \textbf{89.4111} & 108.9197 & 103.7431 & 129.0042 & 131.8705 & 149.9281 \\
\cline{2-8}
  \multirow{5}{*}{A3} & MAB     & 0.0810 & 0.1282 & 0.0839 & 0.1653 & 0.2351 & 0.3761 \\
                      & Var & 0.0166 & 0.0214 & 0.0107 & 0.0047 & 0.0030 & 0.0017 \\
                      & MSE      & \textbf{0.0433} & 0.0577 & 0.0588 & 0.1017 & 0.1427 & 0.2868 \\
                       & CR       & \textbf{7.3705} & 9.5938 & 11.1873 & 20.7544 & 28.9645 & 59.8289 \\
                       & CL     & \textbf{128.1061} & 130.2987 & 131.8753 & 141.5166 & 149.7327 & 180.5282 \\
\bottomrule
\end{tabular}
}
\label{tab:Sim1}
\end{table}

Table \ref{tab:Sim1} presents the mean absolute bias (MAB), mean variance (Var),  mean squared error (MSE), cumulative risk (CR), and cumulative forecast loss (CL) of  VaR forecasts in various settings and the above   approaches. {We report the detailed mean-forecasting results in the Supplementary Material. } The mean absolute bias is computed as $\text{MAB}=\frac{1}{T-t_0+1}\sum_{t=t_0}^T |\frac{1}{n}\sum_{l=1}^n\hat{\theta}_t^{(l)}-\theta_{t}|$, where $\hat{\theta}_t^{(l)}$ is the estimated mean or VaR for the $l$th experiment and $\theta_t$ is the corresponding true parameter at ${t}$. Similarly, the average variance given by $\text{Var}=\frac{1}{T-t_0+1}\sum_{t=t_0}^{T}\frac{1}{n-1}\sum_{l=1}^n(\hat{\theta}_t^{(l)}-\frac{1}{n}\sum_{l=1}^n\hat{\theta}_t^{(l)})^2$  and the mean squared error is $\text{MSE}=\frac{1}{T-t_0+1}\sum_{t=t_0}^{T}\frac{1}{n}\sum_{l=1}^n(\hat{\theta}_t^{(l)}-\theta_t)^2$.  
The cumulative risk over time is defined as the average cumulative excess risk based on the expected loss function $F(\theta)$, given by $$\text{CR}=\frac{1}{n}\sum_{l=1}^n\sum_{t=t_0}^T(F(\hat{\theta}_t^{(l)})-F(\theta_{t})):=\frac{1}{n}\sum_{l=1}^n\sum_{t=t_0}^T\text{CR}_t^{(l)}.$$ 
Particularly,  $\text{CR}_t^{(l)}=(\hat{\theta}_t^{(l)}-\mu_{t})^2$ for the mean forecast, and $$\text{CR}_t^{(l)}=-\alpha \hat{\theta}_t^{(l)}-\E [X_t\id(X_t<\hat{\theta}_t^{(l)})]+\E[X_t\id(X_t<\VaR_t(\alpha))]+\hat{\theta}_t^{(l)}\mathbb{P}(X_t<\hat{\theta}_t^{(l)})$$ for the VaR forecast, where $\text{VaR}_t(\alpha)=\mu_t+\sigma_t\Phi^{-1}(\alpha)$,  with $\Phi^{-1}(\alpha)$ being the $\alpha$-quantile of the standard normal distribution and $\id(\cdot)$ denoting the indicator function.  The average cumulative forecast loss is given by $\text{CL}=\frac{1}{n}\sum_{l=1}^n\sum_{t=t_0}^T\ell(X_t,\hat{\theta}_t^{(l)}).$ These values indicate the overall forecast performance across various approaches.  The following findings are observed from Tables \ref{tab:Sim1}  and   \ref{tab:supp_Sim1} (in the Supplementary Material).

First, across all three settings,  both the mean and VaR forecasts show that the rolling window approach uniformly outperforms the full-window benchmark in terms of MAB, CR, and CL. This is not surprising, as the full window is  optimal only when no structural break is present. Among rolling windows, although using a larger window achieves relatively small variance,  MAB, MSE (for VaR), CR, and CL improve as the rolling-window size decreases from 750 to 250. This illustrates the sensitivity of forecast accuracy to the window size and highlights the importance of adaptive window selection.

Second, for VaR forecasting, BAWS achieves the lowest MSE, CR, and CL across all three discrete-break settings, while its MAB is either the lowest or close to the best-performing benchmark. This indicates the empirical advantages of BAWS in forecasting the tail risk measure VaR. Additional analyses for mean forecasting are included in Section \ref{sec:supp-additional-tables} of the Supplementary Material.
\iffalse
Third, for the mean forecast, BAWS and SAWS consistently produce lower MAB CR, and CL compared to rolling window approaches. 
%BAWS calibrates its threshold via bootstrap, making it data-driven and responsive to the local volatility.  
The advantage of BAWS is more evident in the challenging Setting A3, where the piecewise variance complicates window selection.  In this setting, the bootstrap threshold adjusts to changes in uncertainty and reduces the risk of selecting windows spanning a structural break, whereas the deterministic threshold in SAWS cannot adapt to such changes. In more stable regimes (Settings A1--A2), the two adaptive methods perform comparably.
\fi

Figure~\ref{Fig:A1} showcases the temporal pattern of the estimated mean (left panel) and the estimated VaR (right panel) under Setting A1. 
Plots for Settings A2--A3 are reported in Figures~\ref{Fig:A.4}--\ref{Fig:A.5} of the Supplementary Material. 
In the pre-break regime, all methods produce stable estimates that align closely with the true parameters. After a structural break occurs, the full-window benchmark continues to pool pre-break observations, and fixed rolling windows still mix pre- and post-break data for roughly one window length, thereby delaying adaptation to the new regime. In contrast, the adaptive window selection procedures react more rapidly to the break, with BAWS typically adjusting the selected window more decisively around the change point.

\begin{figure}[t]
    \centering
 \includegraphics[width=\linewidth]{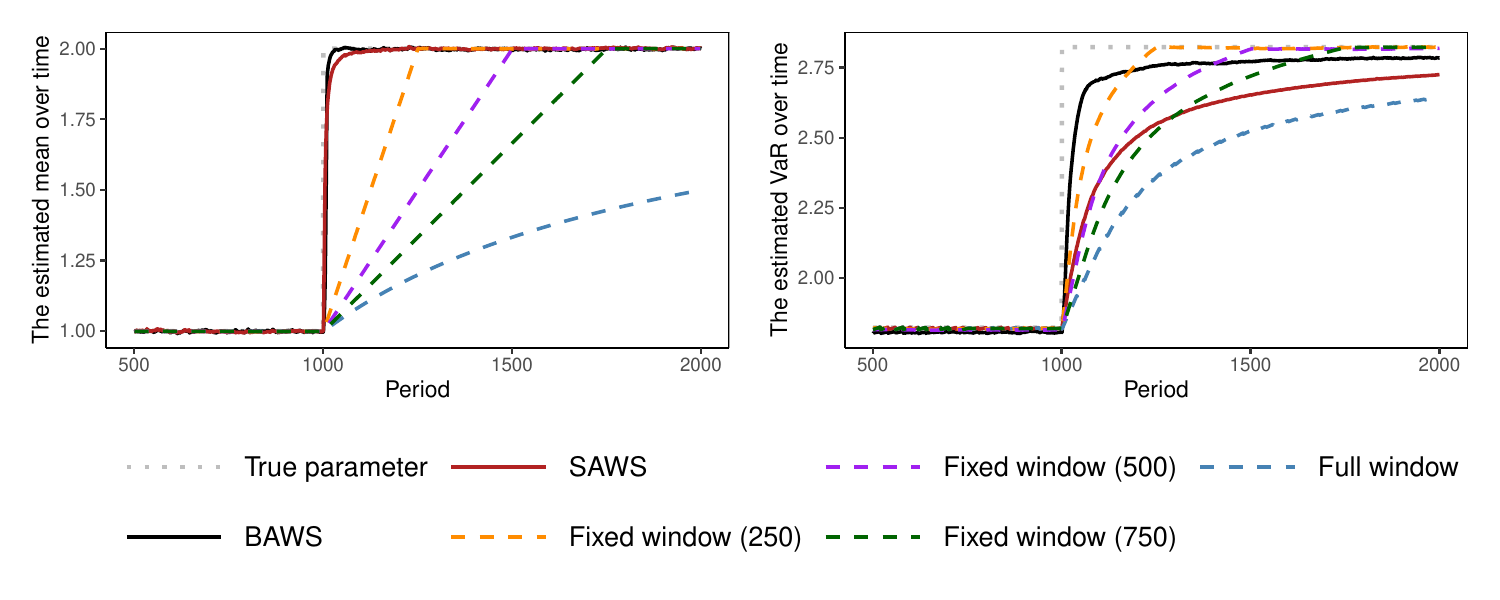}
    \caption{The patterns of mean and VaR estimators over time under Setting A1.}
    \label{Fig:A1}
\end{figure}

\subsection{Scenario 2: Continuous mean shifts}\label{Sec:Sim2}  
We next simulate dynamic mean shifts in the data-generating process, mimicking continuous changes in market trends over time. Although historical and future observations are no longer identically distributed, recent observations remain informative because the distribution changes gradually.

We generate data from $N(\mu_t,\sigma_t^2)$ as in Section \ref{Sec:Sim1}. Parameter settings are as follows:
\begin{itemize}
\item \textbf{Setting B1 (Mean generated by a sine function).}
The mean is \(\mu_t=\sin(2\pi t/T)\), with constant variance 
\(\sigma_t^2=0.25\) for \(t=1,\ldots,T\).
\item \textbf{Setting B2 (Mean generated by a Brownian motion).}
The mean parameter $\mu_t$ changes according to
\(\mu_t-\mu_{t-1}\overset{\text{i.i.d.}}{\sim}N(0,1/T)\), while the variance \(\sigma_t^2=0.25\) for \(t=1,\ldots,T\).
\item \textbf{Setting B3 (Mean generated by a geometric Brownian motion).}
The mean parameter follows
\(
\mu_t=\mu_0 \exp((\mu-\tfrac{1}{2}\sigma^2)\frac{t}{T}
+ \sigma W_t),
\)
where \( W_t-W_{t-1} \overset{\text{i.i.d.}}{\sim} N(0,1/T) \) for \( t=1,\ldots,T \).
We set \( \mu_0=1 \),  \( \mu=0.5 \), \( \sigma^2=\sigma_t^2 = 0.25 \).
\end{itemize}
\begin{table}[!ht]
\caption{The bias, variance, MSE, cumulative risk, and forecast loss for VaR across BAWS, SAWS, fixed windows (250, 500, 750), and full window across Settings B1--B3.}
\scriptsize
\renewcommand{\arraystretch}{0.8}
\centering\resizebox{0.98\textwidth}{!}{
\begin{tabular}{clccccccc}
        \toprule          
          \multirow{2}{*}{VaR} & & \multirow{2}{*}{BAWS} & \multirow{2}{*}{SAWS} &\multicolumn{3}{c}{Fixed Window} & \multirow{2}{*}{Full}\\
  \cline{5-7}
 &&&&250 & 500 & 750 \\
 \midrule
 \multirow{5}{*}{B1} &MAB    &0.1277 
           & 0.8796        & 0.2463             & 0.4801            & 0.6688             & 0.9496           \\
                       &                       Var &0.0137 
           & 0.0046        & 0.0047             & 0.0026            & 0.0020             & 0.0013           \\
                                            & MSE&	\textbf{0.0343} 
 & 	1.0879 &	0.0858& 	0.3287& 	0.6681 &	1.2127           \\
                                             & CR       & \textbf{5.0469} 
           & 52.7609       & 10.1735            & 27.6309           & 41.2991            & 57.8267          \\
                                             & CL   & \textbf{82.4291}           & 130.0980      & 87.5898            & 105.0332          & 118.6737           & 135.1667         \\
                       \cline{2-8}
                        \multirow{5}{*}{B2} & MAB     & 0.1218 
           & 0.2737        & 0.1248             & 0.1668            & 0.1916             & 0.2801           \\
                                             & Var & 0.0077 
           & 0.0025        & 0.0047             & 0.0024            & 0.0017             & 0.0012           \\
                       &                        MSE&	0.0279 
  & 	0.1035& 	\textbf{0.0277}& 	0.0395& 	0.0509& 	0.1075      \\
                       &                       CR       &5.3636 
& 24.9300       & \textbf{4.8871}             & 7.8038            & 10.6236            & 26.0673          \\
                       &                       CL     & 82.7521      & 102.3475      & \textbf{82.3065}            & 85.2086           & 88.0052            & 103.4857         \\
                       \cline{2-8}
                        \multirow{5}{*}{B3} & MAB &0.0836&0.1699&0.0853&0.1230&0.1478&0.1708 
\\

         &Var& 0.0070&0.0012&0.0045&0.0023&0.0016&0.0011 \\
                                             & MSE&	0.0177&0.0446&\textbf{0.0155}& 	0.0229&0.0310&0.0451\\ 

                                             & CR&3.2406&9.3411&\textbf{2.7337} 	&4.1502&5.8178&9.4544 \\
                                            & CL &80.6187&86.7233&\textbf{80.0953} 	&81.5247&83.1921&86.8311\\ 
                       \bottomrule
\end{tabular}
} \label{tab:Sim2}
\end{table}

Table~\ref{tab:Sim2} summarizes the MAB, Var, MSE, CR, and CL for VaR forecasting, while mean-forecasting results are given in Table~\ref{tab:supp_Sim2} of the Supplementary Material. Under Setting B1, BAWS attains the lowest MAB, MSE, CR, and CL for VaR forecasting. Under Setting B2, BAWS remains competitive and achieves performance close to the best fixed-window benchmark. These results indicate that BAWS performs well in environments with cyclic fluctuations or persistent stochastic drift, where the long window is clearly mismatched with the current period and the bootstrap threshold effectively captures such deviation. Under Setting B3, BAWS slightly underperforms the best-performing method in terms of CR and CL, but remains highly comparable.

 %The best-performing methods are SAWS for mean forecasting and the fixed window of 250 for VaR forecasting.
\begin{figure}[!ht]
    \centering
 \includegraphics[width=0.98\linewidth]{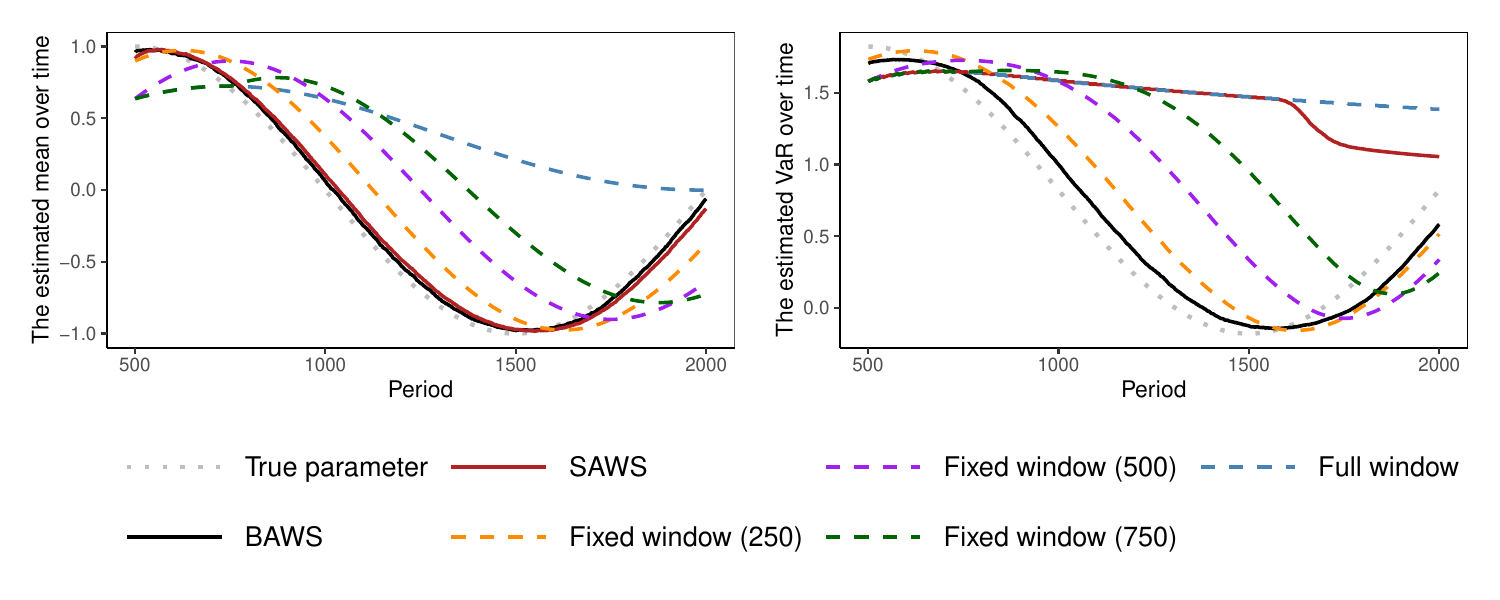}
    \caption{The patterns of mean and VaR estimators over time under Setting B1. }
    \label{Fig:B1}
\end{figure}

Figure~\ref{Fig:B1} illustrates the estimation paths under Setting B1 and is broadly consistent with Table~\ref{tab:Sim2}; the corresponding mean-forecasting summaries are reported in Table~\ref{tab:supp_Sim2} of the Supplementary Material. For VaR forecasting, BAWS tracks the changing target well and performs better than SAWS under the parameter settings of Section~7.1 in \cite{huang2025stability}. Additional trajectory plots for Settings B2--B3 are reported in Figures~\ref{Fig:B2}--\ref{Fig:B3} of the Supplementary Material. Overall, BAWS is more effective under pronounced distributional changes, whereas its gains may diminish under smoother scenarios.

\subsection{Scenario 3: Dynamic volatility shifts}\label{Sec:Sim3}
This study focuses on scenarios where the volatility of the data-generating distribution gradually changes over time.

 Following \cite{hoga2023monitoring} and \cite{wang2025backtesting}, we adopt a skewed-$t$ GARCH(1,1) process:
 \begin{equation*}
     L_t=-\sigma_t\epsilon_t,\quad\sigma_t^2=0.00001+0.04L_{t-1}^2+\gamma_t \sigma_{t-1}^2
 \end{equation*}
and $\{\epsilon_t\}_{t=1}^T$ are i.i.d.~innovations from a skewed Student-$t$ distribution proposed by \cite{fernandez1998bayesian}, with zero mean, unit variance, degrees of freedom $\nu=5$ and skewness parameter $r=0.95$. Let $T=2000$ and $\gamma_t=0.7+0.25\id(t>1000)$, which indicates a structural change of the data-generating process after the midpoint.

\begin{figure}[t]
    \centering
 \includegraphics[width=0.98\linewidth]{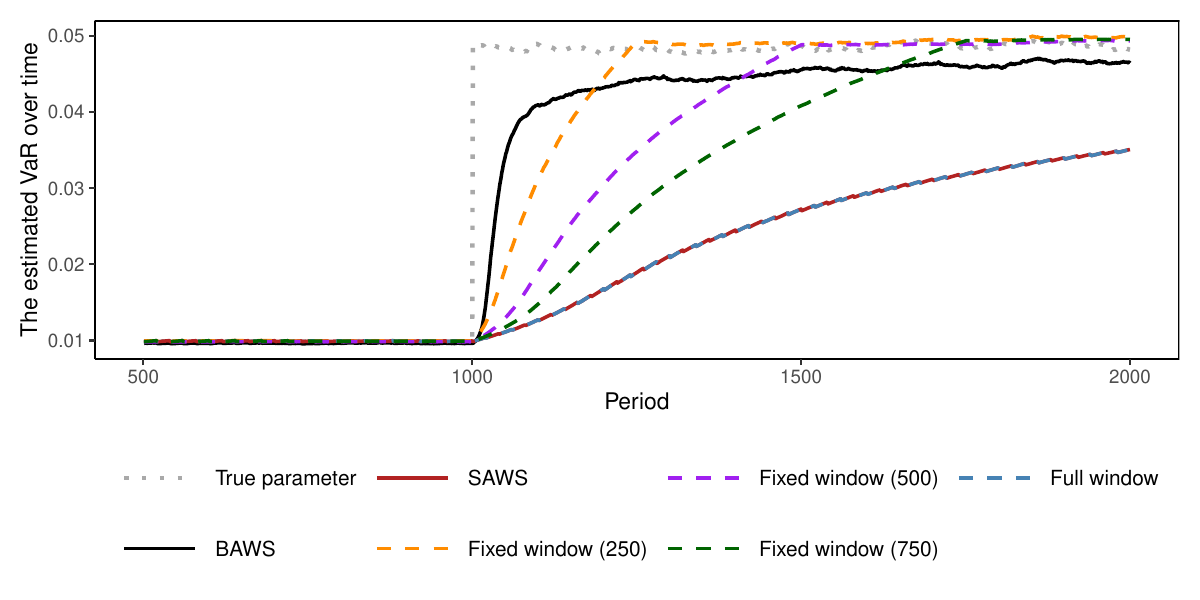}
    \caption{The pattern of VaR forecasts over time under GARCH setting.}
    \label{Fig:C.1}
\end{figure}

\begin{table}[!ht]
\caption{The bias, variance, MSE, cumulative risk, and loss for VaR across BAWS, SAWS, and fixed windows (250, 500, 750), and full window under GARCH.}
\centering
\setlength{\tabcolsep}{4pt}
%\resizebox{0.8\textwidth}{!}{
\begin{tabular}{lcccccc}
\toprule
\multirow{2}{*}{GARCH}& \multirow{2}{*}{BAWS} & \multirow{2}{*}{SAWS} &\multicolumn{3}{c}{Fixed Window} & \multirow{2}{*}{Full}\\
  \cline{4-6}
 &&&250 & 500 & 750 \\
\midrule
MAB     & 0.0034&0.0156&0.0031 &0.0056& 0.0081             & 0.0153  
 \\
Var  & 0.8368 & 0.1771 & 1.5499 & 0.9016&0.5994&	0.1979
 \\
MSE     & \textbf{0.0001}  & 0.0006& 0.0002&0.0003&0.0004&0.0005 
 \\
CR       & \textbf{0.2816} & 1.5481 & 0.4175 & 0.6823&0.9021 &1.5220 
 \\
CL     & \textbf{4.1676} & 5.4367 & 4.2793 & 4.5471&4.7674&5.3874 
 \\
\bottomrule
\end{tabular}
%}
\label{tab:Sim_C}
\vspace{1mm}
\begin{minipage}{0.8\textwidth}
\footnotesize
\textit{Note.}  Variances should be obtained by multiplying the reported numbers by $10^{-4}$.
\end{minipage}
\end{table}
Given that the conditional mean of $L_t$ is zero, we focus on forecasting the  VaR of $L_t$. As mentioned in Section~\ref{sec:boot}, the moving block bootstrap is used to preserve the dependence structure  in time-series data. We therefore implement this bootstrap method with $B=500$  and $l_i=\lceil i^{1/3}\rceil$ . Table~\ref{tab:Sim_C} compares various approaches under this setting. Overall, BAWS achieves the lowest MSE, CR, and CL among all competing approaches, while its MAB is slightly higher than that of the fixed window of 250.  In contrast, fixed-window benchmarks are less responsive to the dynamic change of volatility. The  pre-specified threshold  in SAWS tends to admit overly long windows and leads to a similar result as the full window approach. Figure~\ref{Fig:C.1} further illustrates the dynamic behavior of the forecasts.  BAWS responds quickly to the structural change in volatility dynamics and tracks the true VaR path more closely over time, while the fixed window methods  achieve accurate forecasts only after most pre-break observations have been discarded.

Overall, these results highlight the importance of dynamically adjusting the window size to improve forecast accuracy.  %While the rolling window approachs can accommodate non-stationarity to some extent, they lack the sensitivity to respond promptly to the most recent and substantial  distribution shifts. 
The proposed approach addresses this need by incorporating data characteristics into the threshold. This improvement, however, may slightly compromise statistical efficiency when selecting relatively short windows.

Additional simulation results, including sensitivity analyses with respect to the bootstrap threshold level $\beta$, are reported in Section~\ref{sec:supp-sensitivity} of the Supplementary Material.
\iffalse
Additional mean-forecasting results, including a sensitivity analysis with respect to the bootstrap threshold level $\beta$, are reported in the Supplementary Material. 
The sensitivity results show that BAWS remains stable over a wide range of threshold levels, with larger $\beta$ values leading to longer selected windows.

As an additional robustness check, we examine the sensitivity of BAWS to the bootstrap threshold level $\beta$. Table~\ref{tab:sensitivity-beta-mean} in the Supplementary Material reports results for $\beta\in\{0.80,0.85,0.90,0.95,0.99\}$ under Setting A1 for mean forecasting. The results show that BAWS remains stable across threshold levels: larger $\beta$ values lead to longer selected windows, while the baseline choice $\beta=0.90$ provides a balanced specification.
\fi
\section{Empirical analysis}\label{sec:emp}
In this section, we apply BAWS to a real-world dataset and compare its VaR and ES forecast performance with the SAWS, rolling window, and full window approaches. We analyze the daily losses, defined as the negative log-returns, of the S\&P 500 index from January 4, 2005 to October 30, 2025. As shown in Figure \ref{Fig:real1}, this period involves several significant market fluctuations associated with the 2008 global financial crisis (GFC), COVID-19 pandemic (COVID), and the 2025 U.S. tariff measures (Tariff).

\begin{figure}[!t]
    \centering
    \includegraphics[height=0.4\textheight,width=0.8\linewidth]{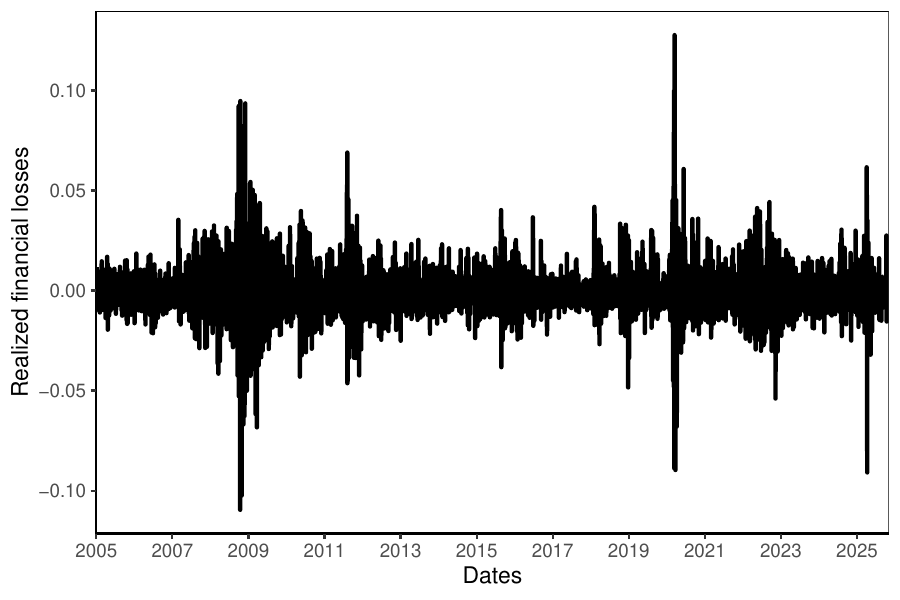}
    \caption{Realized losses of the S\&P 500 index from January 4, 2005 to October 30, 2025.}
    \label{Fig:real1}
\end{figure}

 Using historical data, we perform a rolling estimation of the VaR and ES from December 28, 2006 to October 30, 2025 by employing adaptive windows determined by BAWS and SAWS methods, fixed windows of size $\{250, 500, 750\}$, or full window. We  perform the moving block bootstrap method with $B=1000$  and $l_i=\lceil i^{1/3}\rceil$. Considering the non-stationarity of financial markets, we cap the maximum window of BAWS at 1000 to ensure computational efficiency. We set the threshold level as 0.9. For SAWS, we set $\alpha_\tau=0.1$ and $C_\tau=0.05$, with sensitivity results reported in Supplementary Section~\ref{sec:supp-saws-sensitivity}.

We evaluate VaR and ES forecast performance of various approaches by comparing their average forecast losses over the entire prediction period (2006–2025) as well as across three extreme episodes (the GFC, the COVID-19 crisis, and the 2025 tariff-related episode), as reported in Table \ref{tab:real1}.  For a prediction period \(\mathcal T\), the average forecast loss is given by \(
\sum_{t\in\mathcal T}
S_{V,E,\alpha}(x_t,\widehat{\text{VaR}}_{\alpha}(X_t),\widehat{\text{ES}}_{\alpha}(X_t))/|\mathcal T|,\)
where \(x_t\) is the realized loss and \((\widehat{\text{VaR}}_{\alpha}(X_t),\widehat{\text{ES}}_{\alpha}(X_t))\) is the risk forecast at time \(t\). Over the period,  BAWS attains the lowest average forecast loss over the full evaluation period. The sub-period results further suggest that BAWS remains competitive during the three stress periods, supporting its ability to adapt to changing market conditions.

\begin{figure}[!t]
    \centering
       \includegraphics[height=0.6\textheight, width=0.75\linewidth]{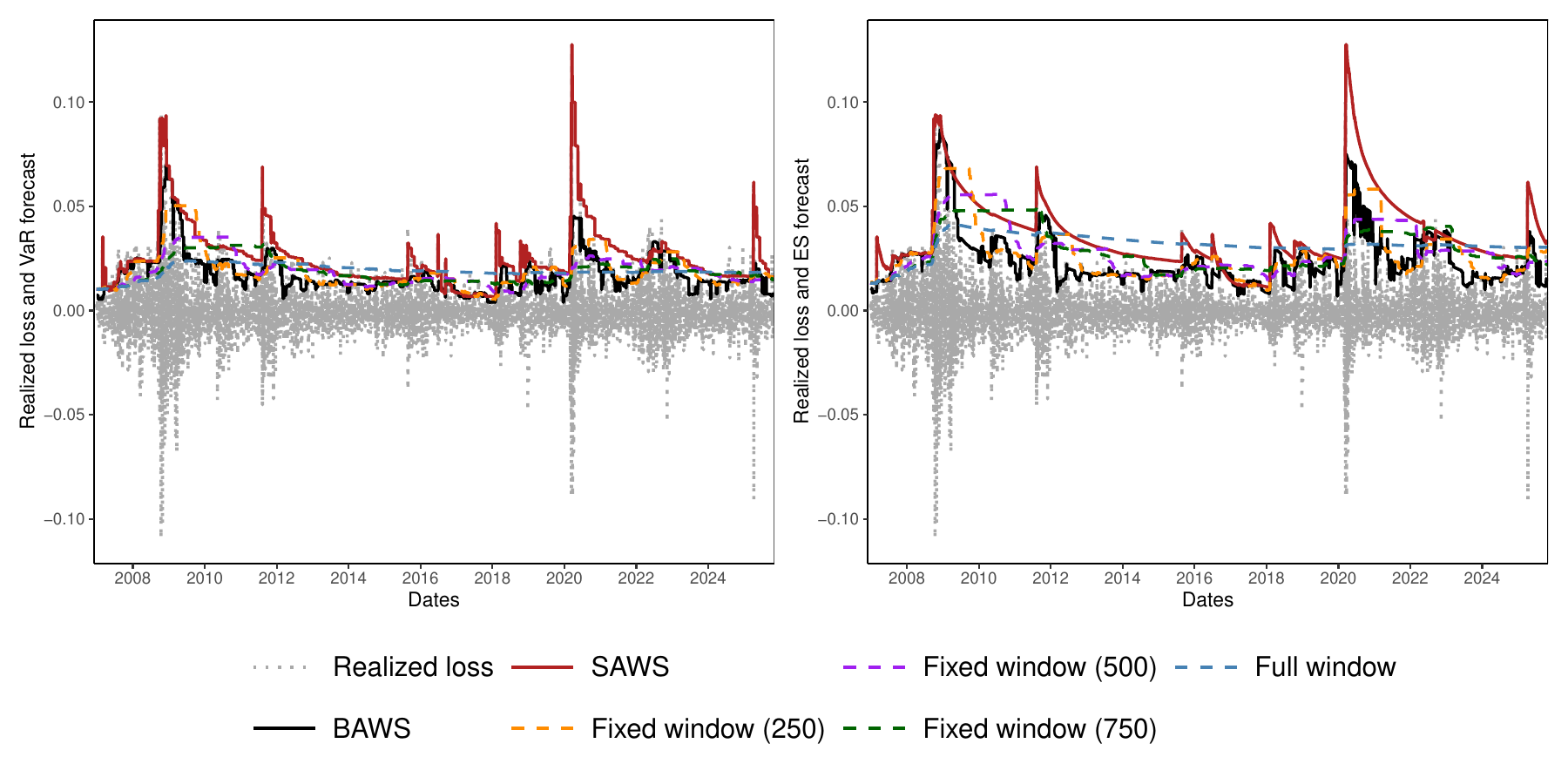}
    \caption{VaR and ES forecasts from December 28, 2006 to October 30, 2025. Left panel: Realized loss and VaR forecast. Right panel: Realized loss and ES forecast.}
    \label{Fig:real2}
\end{figure}

\begin{figure}[!t]
    \centering
\includegraphics[width=0.7\linewidth, height=0.5\textheight]{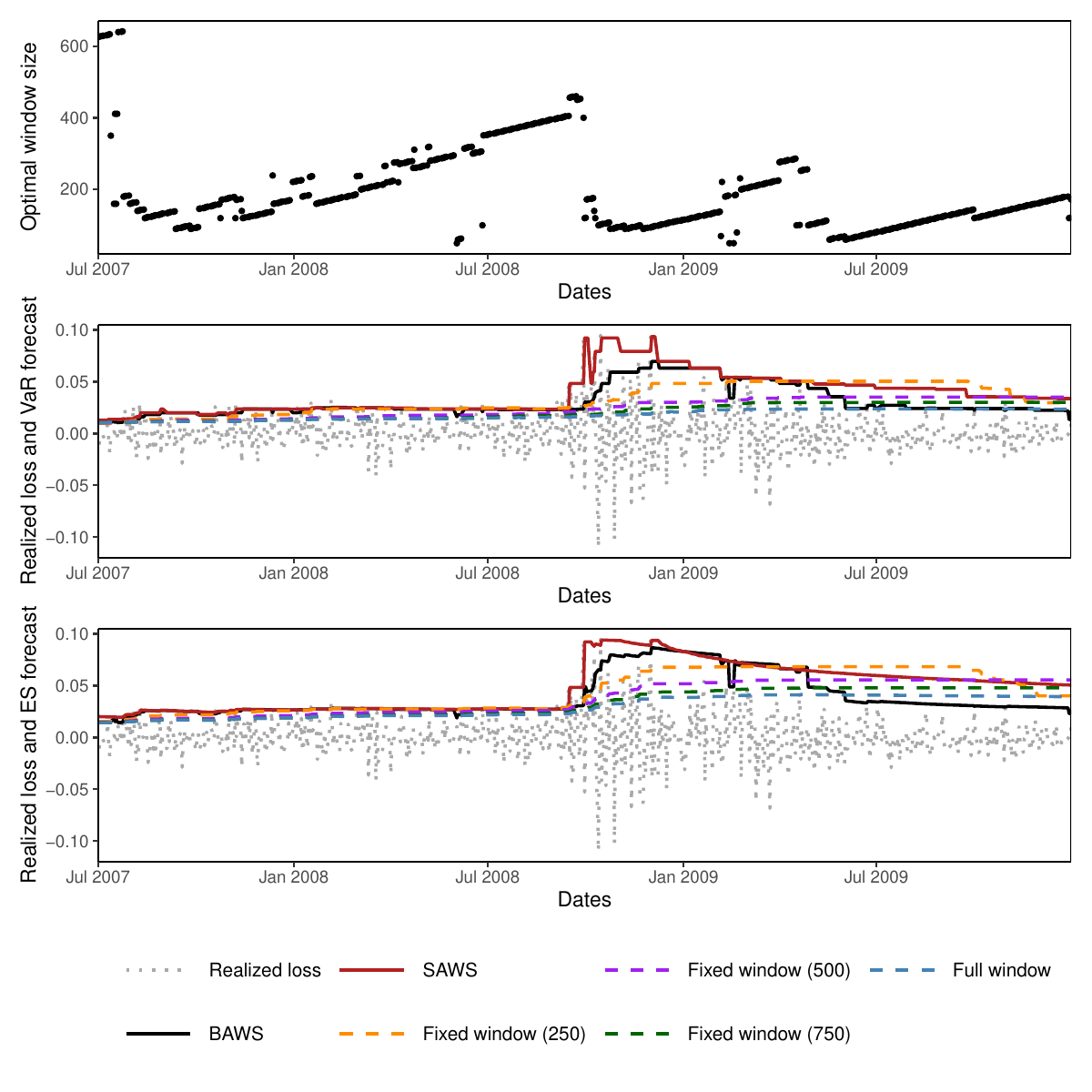} %\includegraphics[height=0.8\linewidth, width=0.8\linewidth]{fig_final/VaR_subplot_2008_block.pdf}
    \caption{Temporal dynamics of selected window sizes and joint VaR--ES forecasts during the 2008
global financial crisis. Top panel: Optimal window. Middle panel: VaR forecasts. Bottom panel: ES forecasts.}
    \label{Fig:real_2008}
\end{figure}
\iffalse
\begin{figure}[!t]
    \centering
        \includegraphics[width=0.7\linewidth, height=0.5\textheight]{fig_final/VaR_subplot_2020_block.pdf}
    \caption{Temporal dynamics of optimal window sizes, VaR, and ES forecasts during COVID-19 pandemic. Top panel: Optimal window. Middle panel: VaR forecasts. Bottom panel: ES forecasts.}
    \label{Fig:real_2020}
\end{figure}
\begin{figure}[!t]
    \centering
        \includegraphics[width=0.7\linewidth, height=0.5\textheight]{fig_final/VaR_subplot_2024_block.pdf}
    \caption{Temporal dynamics of optimal window sizes, VaR, and ES forecasts during the 2025 U.S. tariff measures. Top panel: Optimal window. Middle panel: VaR forecasts. Bottom panel: ES forecasts.}
    \label{Fig:real_2024}
\end{figure}
\fi

Figure \ref{Fig:real2} depicts the temporal evolution of VaR and ES estimates across different methodologies. The left panel presents the dynamics of VaR estimates and the realized loss, while the right panel provides the corresponding ES estimates. These two panels jointly provide a comprehensive view of the forecast performance.  BAWS and SAWS adjust more quickly around stress episodes, whereas longer fixed windows and the full window produce smoother trajectories with more pronounced inertia. In particular, SAWS exhibits sharper spikes in tail-risk forecasts during extreme-loss periods, suggesting an overreaction (overfitting) to transient shocks. 

 Figures \ref{Fig:real_2008} and \ref{Fig:real_2020}--\ref{Fig:real_2024} (the latter in the Supplementary Material) present the optimal window sizes and the corresponding VaR and ES estimates across three extreme events. Figure \ref{Fig:real_2008} covers the period from July 2, 2007, to December 31, 2009, encompassing the 2008 financial crisis. Figure \ref{Fig:real_2020} spans from December 2, 2019, to December 31, 2020, capturing the drastic market volatility during the COVID-19 pandemic.  Figure \ref{Fig:real_2024} corresponds to the period from January 3, 2025 to October 30, 2025, around the implementation of the 2025 tariff policy episode, which began on February 4, 2025, with additional measures introduced in April 2025. When the market undergoes significant changes, the market losses tend to deviate from the historical trend. In such scenarios,  a smaller window is typically preferred to reduce the deviation, as demonstrated in the top panels in Figures \ref{Fig:real_2008}, \ref{Fig:real_2020}, and \ref{Fig:real_2024}. As conditions stabilize, the selected window size tends to expand. Conversely, during periods of market stability, incorporating more data points helps decrease variance in the estimates. The forecast performance depicted in these figures is consistent with the results presented in Table \ref{tab:real1}. The two adaptive window selection approaches outperform most fixed-window and full-window benchmarks overall. 

\begin{table}[t]
\caption{Average forecast loss for VaR--ES joint forecasts across periods for BAWS, SAWS, fixed windows, and full window.}
\begin{threeparttable}
\begin{center}
\resizebox{0.8\textwidth}{!}{
\begin{tabular}{lcccccc}
\toprule
\multirow{2}{*}{Forecast Loss} & \multirow{2}{*}{BAWS} & \multirow{2}{*}{SAWS} &\multicolumn{3}{c}{Fixed Window} & \multirow{2}{*}{Full}\\
  \cline{4-6}
 &&&250 & 500 & 750 \\
\midrule
2006--2025 & \textbf{2.2642} &2.3145 &2.3565&2.4118&2.4174 	&2.4447
 \\
 GFC     & 3.2513&\textbf{3.1852}&3.6631& 3.8319 &4.0414 &	4.1779 
 \\
COVID        & 4.0501&\textbf{3.8752} &4.1750&4.1057&4.2218& 	4.1412
  \\
Tariff      & 2.3124&2.3856&	2.3891& 2.4041&\textbf{2.3021}&2.3292
  \\
\bottomrule
\end{tabular}
}
\vspace{1mm}
\begin{minipage}{0.8\textwidth}
\footnotesize
\textit{Note.} 
The numbers are expressed in percentage (\%).
\end{minipage}
\end{center}
\end{threeparttable}
\label{tab:real1}
\end{table}

\section{Conclusion}\label{sec:con}
This paper develops a bootstrap-based adaptive window selection method for risk forecasting in nonstationary environments. The proposed BAWS procedure is designed for sequentially observed data and applies to both independent and dependent observations. Unlike rolling window approaches, BAWS adaptively updates the historical sample used for forecasting by comparing candidate windows with shorter reference windows using empirical scoring losses. A candidate window is retained only when the loss difference remains below a bootstrap-based threshold, and the selected window is the largest admissible window.

On the theoretical side, we provide an asymptotic justification for the bootstrap threshold and study the selection behavior under squared loss. In a single mean-shift setting, BAWS rejects overly long windows with high probability. 
Simulation studies and an empirical analysis of financial data show that BAWS reduces cumulative forecast loss relative to several benchmark approaches, including fixed-window and full-sample methods, and remains competitive with SAWS. 
%Several directions remain for future research. One is to develop a more comprehensive characterization of multiple pairwise comparisons for overlapping windows. Another is to incorporate uncertainty quantification into the bootstrap threshold to improve out-of-sample performance.

\section*{Data availability statement}

The empirical data are based on historical S\&P 500 Index (\^{}GSPC) prices obtained from Yahoo Finance, available at \url{https://finance.yahoo.com/quote/%5EGSPC/history}. 

\bibliographystyle{apalike}
\bibliography{JASA/reference}

\newpage

\appendix

\counterwithin{equation}{section}
\counterwithin{figure}{section}
\counterwithin{table}{section}
\counterwithin{theorem}{section}   % covers lemma, proposition, etc. via shared counter

% Reformat to use section letter (A, B, ...) as prefix
\renewcommand{\theequation}{S.\arabic{equation}}
\renewcommand{\thefigure}{S.\arabic{figure}}
\renewcommand{\thetable}{S.\arabic{table}}
\renewcommand{\thetheorem}{S.\arabic{theorem}} 
\renewcommand{\theproposition}{S.\arabic{proposition}} 
\renewcommand{\thelemma}{S.\arabic{lemma}} 

  \begin{center}{\Large \bf Supplementary materials for ``Adaptive Window Selection for Financial Risk Forecasting''}
  \end{center}

\section{Additional theoretical results and proofs}\label{sec:app}
\subsection{Proof of Theorem \ref{thm:limit}}\label{app:limit}

    \begin{proof}[Proof of Theorem \ref{thm:limit}]
      Recall that \(\hat{\bm\theta}^{(b)}_n\) denotes the empirical bootstrap estimator, defined as the minimizer of
\[
f_n^{(b)}(\bm\theta)=\frac{1}{n}\sum_{l=1}^n \ell(\bm X_{t-l}^{(b)},\bm\theta),
\]
where \(\bm X_{t-n}^{(b)},\dots,\bm X_{t-1}^{(b)}\) are independent and identically distributed bootstrap resamples drawn from the empirical distribution \(\hat{\mathbb P}_n\), conditional on the observed sample \(\bm X_{t-n},\dots,\bm X_{t-1}\).

Under Conditions (C1) and (C2),  Theorem 5.7 in \cite{van2000asymptotic} gives
\(
\hat{\bm\theta}_n\overset{p}{\to}\bm\theta^*
\) and \(\hat{\bm\theta}^{(b)}_n\overset{p}{\to}\bm\theta^*
\) (the latter conditionally on the observed sample, in probability).

Let
\(
{\bm h}_n=\sqrt n(\hat{\bm\theta}_n-\bm\theta^*).
\)
By Condition~(C4), \({\bm h}_n=O_p(1)\). Hence, for any \(\varepsilon>0\),
there exists \(M<\infty\) such that \(\mathbb P(\|{\bm h}_n\|>M)<\varepsilon\)
for all sufficiently large \(n\). On the compact set \(\{\|\bm h\|\le M\}\),
Condition~(C5) gives
\begin{equation}\label{eq:S_expan1}
n\bigl(f_n(\bm\theta^*+\bm h/\sqrt n)-f_n(\bm\theta^*)\bigr)
=
\bm h^\top\bm Z_n+\tfrac12\bm h^\top\Sigma\bm h+o_p(1),
\end{equation}
uniformly in \(\|\bm h\|\le M\), where
\(
\bm Z_n=n^{-1/2}\sum_{l=1}^n\bm\psi(\bm X_{t-l},\bm\theta^*).
\)
The random quadratic criterion
\(
\bm h\mapsto \bm h^\top\bm Z_n+\tfrac12\bm h^\top\Sigma\bm h
\)
has the unique minimizer \(-\Sigma^{-1}\bm Z_n\). Therefore, applying the argmin
version of the argmax continuous mapping theorem
\citep[Theorem~3.2.2]{vandervaart1996weak} to \eqref{eq:S_expan1} in
Condition~(C5), we obtain
\begin{equation}\label{eq:theta_expan}
\sqrt n(\hat{\bm\theta}_n-\bm\theta^*)
=
-\Sigma^{-1}\bm Z_n+o_p(1).
\end{equation}
By the central limit theorem, $\bm Z_n\overset{d}{\to}  N(0,\Omega)$, so
\[
\sqrt{n}(\hat{\bm\theta}_n-{\bm\theta}^*)
\overset{d}{\to} \bm Z\] where $\bm Z\sim
N(0,\tilde{\Sigma})$ with $\tilde\Sigma=\Sigma^{-1}\Omega\Sigma^{-1}.$

Following Condition (C3), taking a second-order Taylor expansion of the population loss $F$ around \(\bm\theta^*\), 
we obtain
\[
F(\bm\theta)-F(\bm\theta^*)=\tfrac12(\bm\theta-\bm\theta^*)^\top\Sigma(\bm\theta-\bm\theta^*)+o(\|\bm\theta-\bm\theta^*\|^2).
\]
Combining it with $\hat{\bm\theta}_n-\bm\theta^*=O_p(n^{-1/2})$, we further have
\[
n\bigl(F(\hat{\bm\theta}_n)-F(\bm\theta^*)\bigr)=\tfrac12\bigl(\sqrt n(\hat{\bm\theta}_n-\bm\theta^*)\bigr)^\top\Sigma\bigl(\sqrt n(\hat{\bm\theta}_n-\bm\theta^*)\bigr)+o_p(1)\overset{d}{\to}\tfrac12 \bm Z^\top\Sigma \bm Z.
\]

For the bootstrap part, Condition~(C4) similarly gives
\(
\sqrt n\|\hat{\bm\theta}_n^{(b)}-\hat{\bm\theta}_n\|=O_p(1)
\)
conditionally on the observed sample.
 Furthermore, applying the argmin version of the argmax
continuous mapping theorem to \eqref{eq:second_expan2} in Condition (C5)
yields
\begin{equation}\label{eq:theta_b_expan}
\sqrt{n}(\hat{\bm\theta}^{(b)}_n-\hat{\bm\theta}_n)
=
-\Sigma^{-1}\bm Z_n^{(b)}+o_p(1),
\end{equation}
conditionally on the observed sample, in probability.
Recall that
\[
\bar{\bm\psi}_n
=
\frac1n\sum_{l=1}^n
\bm\psi(\bm X_{t-l},\bm\theta^*),
\qquad
\bm Z_n^{(b)}
=
\frac1{\sqrt n}\sum_{l=1}^n
\left(
\bm\psi(\bm X_{t-l}^{(b)},\bm\theta^*)
-
\bar{\bm\psi}_n
\right).
\]
Conditionally on the empirical distribution \(\hat{\mathbb P}_n\), the bootstrap variables
\[
\bm\psi(\bm X_{t-1}^{(b)},\bm\theta^*),\ldots,
\bm\psi(\bm X_{t-n}^{(b)},\bm\theta^*)
\]
are i.i.d.
Therefore,
\begin{align*}
\mathbb E^*\!\left[
\bm\psi(\bm X_{t-l}^{(b)},\bm\theta^*)
\right]
&=
\bar{\bm\psi}_n
\xrightarrow{p}\bm 0,\\
\operatorname{Var}^*\!\left[
\bm\psi(\bm X_{t-l}^{(b)},\bm\theta^*)
\right]
&=
\frac1n\sum_{l=1}^n
\bm\psi(\bm X_{t-l},\bm\theta^*)
\bm\psi(\bm X_{t-l},\bm\theta^*)^\top
-
\bar{\bm\psi}_n\bar{\bm\psi}_n^\top
\xrightarrow{p}\Omega.
\end{align*}
Here, \(\mathbb E^*\) and \(\operatorname{Var}^*\) denote expectation and variance conditional on the observed sample, respectively.
Moreover, for every \(\varepsilon>0\), 
\[
\mathbb E^*\!\left[
\left\|
\bm\psi(\bm X_{t-l}^{(b)},\bm\theta^*)-\bar{\bm\psi}_n
\right\|^2
\mathbf 1\!\left(
\left\|
\bm\psi(\bm X_{t-l}^{(b)},\bm\theta^*)-\bar{\bm\psi}_n
\right\|>\varepsilon\sqrt n
\right)
\right]
\xrightarrow{p}0.
\]
Since \(\bar{\bm\psi}_n\xrightarrow{p}\bm 0\), the above follows from
\[
\frac1n\sum_{l=1}^n
\left\|
\bm\psi(\bm X_{t-l},\bm\theta^*)
\right\|^2
\mathbf 1\!\left(
\left\|
\bm\psi(\bm X_{t-l},\bm\theta^*)
\right\|>\varepsilon\sqrt n/2
\right)
\xrightarrow{p}0,
\]
which is implied by Markov's inequality and the dominated convergence theorem, using
\(
\mathbb E
\left[
\left\|
\bm\psi(\bm X_t,\bm\theta^*)
\right\|^2
\right]<\infty
\)
from Condition~(C5). 
Hence, by the Lindeberg-Feller central limit theorem for triangular arrays, 
see Proposition 2.27 in \cite{van2000asymptotic},  
\[
\mathcal L^*(\bm Z_n^{(b)})\Rightarrow_p N(\bm0,\Omega),
\]
where \(\mathcal L^*(\cdot)\) denotes the conditional law given the observed sample and \(\Rightarrow_p\) denotes weak convergence in probability.
Combining this with
\eqref{eq:theta_b_expan} ensures
\[
\mathcal L^*\!\left(
\sqrt n(\hat{\bm\theta}_n^{(b)}-\hat{\bm\theta}_n)
\right)
\Rightarrow_p
N(0,\tilde\Sigma).
\]
 
 Let $\bm r_n=\sqrt{n}(\hat{\bm\theta}^{(b)}_n-\hat{\bm\theta}_n)$. Since \(\bm h_n=O_p(1)\) and \(\bm r_n=O_p(1)\) conditionally on the observed sample,
for any \(\varepsilon>0\) we may restrict the argument to an event on which
\(\|\bm h_n\|\), \(\|\bm r_n\|\), and \(\|\bm h_n+\bm r_n\|\) are bounded by a fixed constant
with probability at least \(1-\varepsilon\). Applying Condition~(C5) at
\(\bm h=\bm h_n+\bm r_n\) and \(\bm h=\bm h_n\), and subtracting the two expansions, gives 
\[
n\bigl(f_n(\hat{\bm\theta}_n+\bm r_n/\sqrt n)-f_n(\hat{\bm\theta}_n)\bigr)  =
\bm r_n^\top(\bm Z_n+\Sigma{\bm h}_n)
+\frac12 \bm r_n^\top\Sigma\bm r_n
+o_p(1),
\]
uniformly for bounded \(\bm r_n\). Using \eqref{eq:theta_expan},
we have \(\bm Z_n+\Sigma{\bm h}_n=o_p(1)\), and hence 
\[
n\bigl(f_n(\hat{\bm\theta}_n^{(b)})-f_n(\hat{\bm\theta}_n)\bigr)
=
\tfrac12
\bigl(\sqrt n(\hat{\bm\theta}_n^{(b)}-\hat{\bm\theta}_n)\bigr)^\top
\Sigma
\bigl(\sqrt n(\hat{\bm\theta}_n^{(b)}-\hat{\bm\theta}_n)\bigr)
+o_p(1),
\]
conditionally on the observed sample, in probability.

By the continuous mapping theorem,
\[
\mathcal L^*\left(
n\left(f_n(\hat{\bm\theta}^{(b)}_n)-f_n(\hat{\bm\theta}_n)\right)
\right)
\Rightarrow_p
\mathcal L\left(
\tfrac12 \bm Z^\top\Sigma \bm Z
\right).
\]
 Thus, $n(f_n(\hat{\bm\theta}^{(b)}_n)-f_n(\hat{\bm\theta}_n))$ and $n(F(\hat{\bm\theta}_n)-F({\bm\theta}^*))$ have the same limiting distribution in the above sense.
\end{proof}

\subsection{Verification for non-smooth scoring functions}\label{app:verify}
This subsection verifies Conditions (C1)--(C5) for the two non-smooth scoring functions used in the risk-forecasting application in Section~\ref{sec:risk_fore}: the quantile check function for VaR and the Fissler-Ziegel (FZ)  score for the joint $(\VaR_\alpha,\ES_\alpha)$. We focus on the case $G(x)=x$ in \eqref{eq:score_VaR} and on $G_1(x)=x$, $G_2(x)=-e^{-x}/(1+e^{-x})$ in \eqref{eq:score_joint}, but the arguments below can extend to other choices of scoring functions satisfying the conditions in \citet{fissler2016higher}.

\subsubsection{The check function for VaR}
Let $\ell(x,v) = (\one\{x<v\}-\alpha)(v-x)$, and let $F_X$ denote the c.d.f.\ of $X_t$. Assume
\begin{itemize}
\item[(V1)] \(F_X\) has a density \(f_X\) that is continuous in a neighborhood of
\(v^*=\VaR_\alpha(X_t)\), with \(f_X(v^*)>0\);
\item[(V2)] \(\E[|X_t|]<\infty\);
\item[(V3)] The parameter space \(\Theta_V\subseteq\mathbb R\) is compact and contains
\(v^*=\VaR_\alpha(X_t)\) as an interior point.
\end{itemize}
The domain in the VaR equation \eqref{eq:VaR} is localized to \(\Theta_V\) for the theoretical verification. Since the global minimizer is unique and lies in \(\operatorname{int}(\Theta_V)\), this restriction does not change the target.
\paragraph*{Conditions (C1) and (C2).}
Condition (C1) follows from (V1) and (V3). Uniform consistency over the compact
\(\Theta_V\) follows from the Glivenko-Cantelli property of
\(\{\ell(\cdot,v):v\in\Theta_V\}\), since this class is generated by indicator
functions and linear functions with an integrable envelope under (V2).

\paragraph*{Condition (C3).} Direct computation gives that the population loss satisfies
\[
F(v) = \alpha\int_v^\infty(x-v)\,\d F_X(x) + (1-\alpha)\int_{-\infty}^v(v-x)\,\d F_X(x),
\]
so $F'(v) = F_X(v) - \alpha$ and $F''(v)=f_X(v)$. Hence $F'(v^*)=0$, $\Sigma=F''(v^*)=f_X(v^*)>0$.

\paragraph*{Condition (C4).}
The Bahadur representation for the empirical quantile gives
\[
\sqrt n(\hat v_n-v^*)
=
-\frac{1}{f_X(v^*)}
\frac1{\sqrt n}\sum_{l=1}^n
\left(\one(X_{t-l}<v^*)-\alpha\right)
+o_p(1).
\]
Therefore \(
\sqrt n|\hat v_n-v^*|=O_p(1)
\) follows from the central limit theorem. 
The bootstrap analogue,
\[
\sqrt n|\hat v_n^{(b)}-\hat v_n|=O_p(1),
\]
conditionally on the observed sample, in probability, follows from the bootstrap version
of the quantile Bahadur expansion, or more generally from the standard bootstrap
theory for sample quantiles \citep{bickel1981asymptotic}.

\paragraph*{Condition (C5).}
Take
\(
\psi(x,v^*)=\one\{x<v^*\}-\alpha\)
and \(
Z_n=\frac1{\sqrt n}\sum_{l=1}^n\psi(X_{t-l},v^*),
\)
so that \(\E[\psi(X_t,v^*)]=0\) and
\(
\Omega=\Var(\one\{X_t<v^*\})=\alpha(1-\alpha).
\)
By Knight's equality \citep{knight1998limiting}, for \(u=h/\sqrt n\),
\begin{equation}\label{eq:knight}
\ell(x,v^*+u)-\ell(x,v^*)
=
u(\one(x<v^*)-\alpha)
+
\int_0^u
(\one(x<v^*+s)-\one(x<v^*))\,ds.
\end{equation}
Therefore, uniformly for \(|h|\le M\),
\[
n(f_n(v^*+h/\sqrt n)-f_n(v^*))
=
hZ_n
+
n\int_0^{h/\sqrt n}
(F_X(v^*+s)-F_X(v^*))\,ds
+
o_p(1).
\]
Since \(f_X\) is continuous at \(v^*\),
\[
n\int_0^{h/\sqrt n}
(F_X(v^*+s)-F_X(v^*))\,ds
=
\frac12 f_X(v^*)h^2+o(1)
\]
uniformly for \(|h|\le M\). Hence
\[
\sup_{|h|\le M}
\left|
n(f_n(v^*+h/\sqrt n)-f_n(v^*))
-
hZ_n
-
\frac12 f_X(v^*)h^2
\right|
\overset{p}{\to}0.
\]
This verifies the first expansion in Condition (C5). The bootstrap counterpart follows from the same Knight-identity argument applied
conditionally on the observed sample. In this argument, the population distribution is
replaced by the empirical distribution, and the linear term is replaced by the centered
bootstrap score \(Z_n^{(b)}\).

\subsubsection{The Fissler-Ziegel score for (VaR, ES)}
Let \(\bm\theta=(v,e)^\top\), and let
\(\ell(x,v,e)=S_{V,E,\alpha}(x,v,e)\) be the Fissler-Ziegel score in
\eqref{eq:score_joint}, with \(G_1(x)=x\) and
\(
G_2(x)=-\exp(-x)/(1+\exp(-x)).
\)
Let \(\bm\theta^*=(v^*,e^*)^\top\), where
\(v^*=\VaR_\alpha(X_t)\) and \(e^*=\ES_\alpha(X_t)\). Assume
\begin{itemize}
\item[(E1)] The distribution function \(F_X\) has a density \(f_X\) that is continuous
in a neighborhood of \(v^*\), with \(f_X(v^*)>0\);
\item[(E2)] \(\E[X_t^2]<\infty\).
\item[(E3)] The parameter space \(\Theta_{VE}\subseteq\mathbb R^2\) is compact and
\(\bm\theta^*=(v^*,e^*)^\top\in\operatorname{int}(\Theta_{VE})\).
\end{itemize}

\paragraph*{Conditions (C1) and (C2).}
The strict consistency of the Fissler-Ziegel score implies that \(\bm\theta^*\) is the
unique minimizer of the population loss. Together with (E3), this verifies
Condition~(C1). Uniform consistency
over compact \(\Theta_{VE}\) follows from the Glivenko-Cantelli property of the corresponding
score class. This class consists of smooth terms in \((v,e)\) and threshold indicator terms
of the form \(\one\{x<v\}\), and is dominated by an integrable envelope under (E2).

\paragraph*{Condition (C3).}
\iffalse
Under (E1)--(E3), the population loss is twice continuously differentiable in a
neighborhood of \(\bm\theta^*\). The first-order condition satisfies
\(
\nabla F(\bm\theta^*)=\bm0.
\)
Moreover, standard calculations for the Fissler-Ziegel score give
\[
\Sigma=\nabla^2 F(\bm\theta^*)
=
G_2'(e^*)
\begin{pmatrix}
f_X(v^*)/(1-\alpha) & 0\\
0 & 1
\end{pmatrix},
\]
under the present loss convention. Since
\[
G_2'(x)=\frac{e^{-x}}{(1+e^{-x})^2}>0
\]
and \(f_X(v^*)>0\), \(\Sigma\) is positive definite. This verifies Condition~(C3).
\fi
Let \(A(v)=\mathbb E\left[(v-X_t)1\{X_t\ge v\}\right].\)
Since $F_X$ is continuous at $v^*$, we have $F_X(v^*)=\alpha$. Also, under the upper-tail loss convention,
\[
A(v^*)=(1-\alpha)(v^*-e^*).
\]
A direct calculation gives
\[
\frac{\partial F(v,e)}{\partial v}=
(F_X(v)-\alpha)
\left(
1-\frac{G_2(e)}{1-\alpha}
\right),
\]
and
\[
\frac{\partial F(v,e)}{\partial e}
=
G_2'(e)
\left(
e-v+\frac{A(v)}{1-\alpha}
\right).
\]
Therefore $\nabla F(\theta^*)=0$. Furthermore,
\[
\Sigma = \nabla^2 F(\theta^*)
=
\begin{pmatrix}
f_X(v^*)\left(1-\dfrac{G_2(e^*)}{1-\alpha}\right) & 0\\[1.5ex]
0 & G_2'(e^*)
\end{pmatrix}.
\]
For the choice
\[
G_2(x)=-\frac{\exp(-x)}{1+\exp(-x)},
\]
we have $G_2(e^*)<0$ and $G_2'(e^*)>0$. Since $f_X(v^*)>0$, the matrix $\Sigma$ is positive definite. Hence Condition (C3) holds.
\paragraph*{Condition (C4).}
For the joint (VaR, ES) estimator based on the Fissler-Ziegel score, the standard
asymptotic theory gives the linear expansion
\[
\sqrt n
\left\{
(\hat v_n,\hat e_n)^\top-(v^*,e^*)^\top
\right\}
=
-\Sigma^{-1}
\frac1{\sqrt n}\sum_{l=1}^n
\bm\psi(X_{t-l},\bm\theta^*)
+o_p(1),
\]
where \(\bm\psi\) is the generalized score defined below. Since
\(\E[\bm\psi(X_t,\bm\theta^*)\bm\psi(X_t,\bm\theta^*)^\top]<\infty\), the central limit
theorem implies that the leading term is \(O_p(1)\). Hence
\[
\sqrt n\|(\hat v_n,\hat e_n)^\top-(v^*,e^*)^\top\|=O_p(1).
\]

The empirical bootstrap estimator satisfies the analogous conditional linear expansion
\[
\sqrt n
\left\{
(\hat v_n^{(b)},\hat e_n^{(b)})^\top-(\hat v_n,\hat e_n)^\top
\right\}
=
-\Sigma^{-1}
\bm Z_n^{(b)}
+o_p(1),
\]
conditionally on the observed sample, in probability, where \(\bm Z_n^{(b)}\) is the
centered bootstrap score defined in Condition~(C5). Consequently,
\[
\sqrt n\|(\hat v_n^{(b)},\hat e_n^{(b)})^\top-(\hat v_n,\hat e_n)^\top\|=O_p(1).
\]
 This verifies Condition~(C4). 

\paragraph*{Condition (C5).}
\iffalse
The generalized score can be taken as the Fissler-Ziegel identification function
evaluated at \(\bm\theta^*\):
\[
\bm\psi(x,\bm\theta^*)
=
\begin{pmatrix}
\one\{x<v^*\}-\alpha \\[2pt]
G_2'(e^*)
\left(
e^*-v^*
-\frac{1}{1-\alpha}(v^*-x)\one\{x\ge v^*\}
\right)
\end{pmatrix}.
\]
Under (E1)--(E2),
\[
\E[\bm\psi(X_t,\bm\theta^*)]=\bm0,
\qquad
\E[\bm\psi(X_t,\bm\theta^*)\bm\psi(X_t,\bm\theta^*)^\top]<\infty.
\]

To verify the local quadratic expansion, decompose the Fissler-Ziegel score into a
threshold component involving \(\one\{x<v\}\) and a component that is smooth in
\((v,e)\). The threshold component is handled by the same Knight-identity argument as
for the check loss, while the smooth component is handled by a Taylor expansion around
\(\bm\theta^*\). Combining the two expansions gives, for every fixed \(M>0\),
\[
\sup_{\|\bm h\|\le M}
\left|
n\{f_n(\bm\theta^*+\bm h/\sqrt n)-f_n(\bm\theta^*)\}
-\bm h^\top\bm Z_n
-\frac12\bm h^\top\Sigma\bm h
\right|
\overset{p}{\to}0.
\]
The bootstrap counterpart follows from the same decomposition applied conditionally on
the observed sample, with the empirical distribution replacing the population distribution
and \(\bm Z_n\) replaced by the centered bootstrap score \(\bm Z_n^{(b)}\). This verifies
Condition~(C5).
\fi
Define
\[
c^* = 1-\frac{G_2(e^*)}{1-\alpha}.
\]
The generalized score appearing in the local quadratic expansion is
\[
\bm\psi(x,\bm\theta^*)
=
\begin{pmatrix}
c^*(1\{x<v^*\}-\alpha)\\[1.5ex]
G_2'(e^*)
\left[ e^*-v^* +
\frac{1}{1-\alpha}(v^*-x)1\{x\ge v^*\}
\right]
\end{pmatrix}.
\]
Then
\[
\mathbb E[\bm\psi(X_t,\bm\theta^*)]=0.
\]
Indeed, $\mathbb E[1\{X_t<v^*\}-\alpha]=0$, and
\[
\mathbb E\left[
e^*-v^* + \frac{1}{1-\alpha}(v^*-X_t)1\{X_t\ge v^*\}
\right] = e^*-v^*+v^*-e^* = 0.
\]
Moreover, by (E2),
\[
\Omega =
\mathbb E[\bm\psi(X_t,\bm\theta^*)\bm\psi(X_t,\bm\theta^*)^\top]
< \infty.
\]
Let
\[
\bm Z_n =
\frac{1}{\sqrt n}\sum_{l=1}^n \bm\psi(X_{t-l},\bm\theta^*).
\]
The threshold part of the Fissler-Ziegel score is handled by the same
Knight-identity argument applied to
\((\one\{x<v\}-\alpha)(v-x)\), as shown in \eqref{eq:knight}, and by the analogous
identity for \(g(v):=\one\{x\ge v\}(v-x)\),
\[
g(v^*+u)-g(v^*)
=
u\one\{x\ge v^*\}
+
\int_0^u
\left(
\one\{x\ge v^*+s\}
-
\one\{x\ge v^*\}
\right)\,ds .
\]
The remaining terms are smooth in $(v,e)$ and are handled by Taylor expansion around $\bm\theta^*$. Since $f_X$ is continuous at $v^*$, the deterministic second-order term is uniformly approximated by the Hessian $\Sigma$ above. Hence, for every fixed $M>0$,
\[
\sup_{\|\bm h\|\le M}
\left|
n\left(
f_n\left(\bm\theta^*+\frac{\bm h}{\sqrt n}\right)-f_n(\bm\theta^*)
\right) -\bm h^\top \bm Z_n
-\frac{1}{2}\bm h^\top\Sigma \bm h
\right|
\overset{p}{\to} 0.
\]
For the bootstrap expansion, define
\[
\bar{\bm\psi}_n
=
\frac{1}{n}\sum_{j=1}^n\bm\psi(X_{t-j},\bm\theta^*),
\qquad
\bm Z_n^{(b)} =
\frac{1}{\sqrt n}\sum_{l=1}^n
\left(
\bm\psi(X_{t-l}^{(b)},\bm\theta^*)-\bar{\bm \psi}_n
\right).
\]
The conditional version of the preceding Knight-identity and Taylor-expansion
argument, with the empirical distribution replacing the population distribution, gives 
\[
\sup_{\|\bm h\|\le M}
\left|
n\left(
f_n^{(b)}\left(\hat{\bm\theta}_n+\frac{\bm h}{\sqrt n}\right) - f_n^{(b)}(\hat{\bm\theta}_n)
\right)
-\bm h^\top \bm Z_n^{(b)}
-\frac{1}{2}\bm h^\top\Sigma \bm h \right|
\overset{p}{\to} 0
\]
conditionally on the observed sample, in probability. This verifies Condition (C5).

\subsection{Bootstrap validity for dependent data}\label{sec:mbb}
Section~\ref{subsec:iid_bootstrap} treated the case of i.i.d.\ observations and the empirical bootstrap. For the GARCH scenario in Section~\ref{Sec:Sim3} and the empirical analysis in Section~\ref{sec:emp}, BAWS uses the moving block bootstrap (MBB) of \citet{kunsch1989jackknife} with block length $l_n=c\lceil n^{1/3}\rceil$. This subsection extends the asymptotic justification of the bootstrap threshold to the setting of dependent and stationary data.

We assume that the observations within the window $\{t-n,\dots,t-1\}$ are stationary and weakly dependent. Concretely, we replace Conditions (C2), (C4), and (C5) by their dependent-data analogues, while retaining Conditions (C1) and (C3).
\begin{itemize}
\item {\it (M1) Stationarity and mixing.}
The process \(\{\bm X_l\}\) is strictly stationary and \(\beta\)-mixing, with mixing
coefficients
\[
\beta(j)
=
\sup_{t\ge1}
\E\left[
\sup_{B\in\mathcal F_{t+j}^{\infty}}
\left|
\mathbb P(B\mid \mathcal F_{1}^{t})-\mathbb P(B)
\right|
\right],
\qquad
\mathcal F_a^b=\sigma(\bm X_l:a\le l\le b),
\]
satisfying
\(
\sum_{j\ge1}j^{2}\beta(j)^{\delta/(2+\delta)}<\infty
\)
for some \(\delta>0\), and
\(
\E[\|\bm\psi(\bm X_t,\bm\theta^*)\|^{2+\delta}]<\infty,
\)
where \(\bm\psi\) is the generalized score defined in Condition~(C5).
\item {\it (M2) Block length.} The block length $l_n$ used in the MBB satisfies $l_n\to\infty$ and $l_n/n^{1/2}\to 0$ as $n\to\infty$.
\end{itemize}
The choice $l_n=c\lceil n^{1/3}\rceil$ used in this paper satisfies Condition~(M2). Condition~(M1) is satisfied by many GARCH processes under standard regularity and
moment conditions; see \citet{carrasco2002mixing}. 
\begin{theorem}[Bootstrap validity under mixing]\label{thm:mbb}
Suppose Conditions~(C1) and (C3) hold, and the analogue of Conditions~(C4) and (C5) hold under the moving block bootstrap, with  the long-run variance
\[
\Omega_{\mathrm{LR}}
:=
\sum_{j\in\mathbb Z}
\E\!\left[
\bm\psi(\bm X_t,\bm\theta^*)
\bm\psi(\bm X_{t+j},\bm\theta^*)^\top
\right]
\]
where
the right-hand side does not depend on \(t\) by stationarity.
If, in addition, Conditions~(M1)--(M2) hold, then
\[
\mathcal L^*\!\left(
\sqrt n(\hat{\bm\theta}_n^{(b)}-\hat{\bm\theta}_n)
\right)
\Rightarrow_p
N(\bm0,\tilde\Sigma_{\mathrm{LR}}),
\]
where $\tilde\Sigma_{\mathrm{LR}}=\Sigma^{-1}\Omega_{\mathrm{LR}}\Sigma^{-1}$, and
\[
\mathcal L^*\!\left(n\bigl(f_n(\hat{\bm\theta}_n^{(b)})-f_n(\hat{\bm\theta}_n)\bigr)\right)\Rightarrow_p\mathcal L\!\left(\frac{1}{2} Z_{\mathrm{LR}}^\top \Sigma Z_{\mathrm{LR}}\right),
\]
with $Z_{\mathrm{LR}}\sim N(0,\tilde\Sigma_{\mathrm{LR}})$.
\end{theorem}
\begin{proof}[Proof]
Under Condition~(M1), the central limit theorem for stationary mixing sequences \citep{doukhan1994mixing} gives
\[
\bm Z_n
=
n^{-1/2}\sum_{l=1}^n
\bm\psi(\bm X_{t-l},\bm\theta^*)
\Rightarrow
N(\bm0,\Omega_{\mathrm{LR}}).
\]
Together with the dependent-data analogue of Conditions~(C4) and (C5), and the positive definiteness of \(\Sigma\) from
Condition~(C3), the argmin version of the argmax continuous mapping
theorem \citep[Theorem~3.2.2]{vandervaart1996weak} yields
\[
\sqrt n(\hat{\bm\theta}_n-\bm\theta^*)
=
-\Sigma^{-1}\bm Z_n+o_p(1).
\]
Under Conditions~(M1)--(M2), the moving block bootstrap consistently approximates the
distribution of the generalized-score sum, see, e.g., \citet{kunsch1989jackknife} and \citet{buhlmann1995blockwise}. Hence,
\(
\mathcal L^*(\bm Z_n^{(b)})\Rightarrow_p N(\bm0,\Omega_{\mathrm{LR}}),
\)
so the corresponding conditional expansion of $\sqrt n(\hat{\bm\theta}_n^{(b)}-\hat{\bm\theta}_n)$ gives
\[
\mathcal L^*\!\left(
\sqrt n(\hat{\bm\theta}_n^{(b)}-\hat{\bm\theta}_n)
\right)
\Rightarrow_p
N(\bm0,\tilde\Sigma_{\mathrm{LR}}).
\]

Finally, applying the quadratic expansion of \(f_n\) in
Condition~(C5) at \(h=\sqrt n(\hat{\bm\theta}_n^{(b)}-\hat{\bm\theta}_n)\) and \(\bm h=\sqrt n(\hat{\bm\theta}_n-\bm\theta^*)\), and subtracting these two equations,  yields
\[
\mathcal L^*\!\left(
n\{f_n(\hat{\bm\theta}_n^{(b)})-f_n(\hat{\bm\theta}_n)\}
\right)
\Rightarrow_p
\mathcal L\!\left(
\frac12\bm Z_{\mathrm{LR}}^\top\Sigma\bm Z_{\mathrm{LR}}
\right),
\]
as claimed.
\end{proof}
Theorem~\ref{thm:mbb} provides the asymptotic basis for the MBB bootstrap threshold
under stationarity and weak dependence. In particular, by resampling consecutive blocks, the moving block bootstrap retains the
local dependence structure in the data and hence in the generalized score process.
\subsection{Proofs of Theorem \ref{thm:reject_large_window} and Corollary \ref{cor:selected_window}}
\begin{proof}[Proof of Theorem \ref{thm:reject_large_window}]
Let $\tau(t,k_0)$ be the threshold used in the stability test.
Recall that $\hat k_t$ is the largest admissible window. Then,
\[
\{T_k=0\}=
\bigcap_{i<k}\left\{ f_{t,i}(\hat\mu_{t,k})-f_{t,i}(\hat\mu_{t,i})\le \tau(t,i)\right\}.
\]
In particular, taking $i=k_0$ yields
\(
\{T_k=0\}\subseteq
\left\{ f_{t,k_0}(\hat\mu_{t,k})-f_{t,k_0}(\hat\mu_{t,k_0})\le \tau(t,k_0)\right\}.
\)
Equivalently,
\(
\left\{ 
f_{t,k_0}(\hat\mu_{t,k})
-
f_{t,k_0}(\hat\mu_{t,k_0})
> \tau(t,k_0)
\right\}
\subseteq
\{T_k=1\},
\) 
hence
\begin{equation}
\label{eq:lower_bound_reject}
\mathbb P(T_k=1\mid H_1)\ \ge\
\mathbb P\!\left(
f_{t,k_0}(\hat\mu_{t,k})-f_{t,k_0}(\hat\mu_{t,k_0})>\tau(t,k_0)\ \Big|\ H_1
\right).
\end{equation}
Under squared loss, for any $\mu$ we have the identity
\(f_{t,k_0}(\mu)-f_{t,k_0}(\bar X_{t,k_0})=(\mu-\bar X_{t,k_0})^2,
\)
therefore
\begin{equation}
\label{eq:diff_equals_square}
f_{t,k_0}(\hat\mu_{t,k})-f_{t,k_0}(\hat\mu_{t,k_0})
=\big(\bar X_{t,k}-\bar X_{t,k_0}\big)^2.
\end{equation}

Now decompose the $k$-window mean under $H_1$ into the pre- and post-break parts:
\[
\bar X_{t,k}
=\frac{k-k_0}{k}\bar X_1
+\frac{k_0}{k}\bar X_{2},
\qquad
\bar X_{t,k_0}=\bar X_{2}
\]
with $\bar X_1=\frac{1}{k-k_0}\sum_{l=t-k}^{t-k_0-1}X_l$ and $\bar X_{2}=\frac{1}{k_0}\sum_{l=t-k_0}^{t-1}X_l$.
Hence,
\[
\bar X_{t,k}-\bar X_{t,k_0}
=\frac{k-k_0}{k}(\bar X_{1}-\bar X_{2}).
\]
By the weak law of large numbers, as $t\to\infty$, $k\to\infty$, and $k_0\to\infty$,
\[
\bar X_{1}\xrightarrow{p}\mu_1,\qquad \bar X_{2}\xrightarrow{p}\mu_2,
\]
and since $\frac{k-k_0}{k}\to c$, we obtain
\[
\big(\bar X_{t,k}-\bar X_{t,k_0}\big)^2\xrightarrow{p}c^2(\mu_1-\mu_2)^2\qquad\text{under}~H_1.
\]
Recall that
$\mathbb P(\tau(t,k_0)\le c^2(\mu_1-\mu_2)^2-\epsilon\mid H_1)\to 1$. For the given $\epsilon$, we have 
\[\mathbb P\left(\big(\bar X_{t,k}-\bar X_{t,k_0}\big)^2>c^2(\mu_1-\mu_2)^2-\frac{\epsilon}2\mid H_1\right)\to 1.\]
Therefore,
\begin{small}
\begin{equation*}
    \mathbb P\left(\left\{\big(\bar X_{t,k}-\bar X_{t,k_0}\big)^2>c^2(\mu_1-\mu_2)^2-\epsilon\right\}\cap\left\{\tau(t,k_0)\le c^2(\mu_1-\mu_2)^2-\epsilon\right\}\mid H_1\right)\to 1.
\end{equation*}
\end{small}
It follows that
\[
\mathbb P\!\left(\big(\bar X_{t,k}-\bar X_{t,k_0}\big)^2>\tau(t,k_0)\ \Big|\ H_1\right)\to 1.
\]
Using  \eqref{eq:lower_bound_reject} and \eqref{eq:diff_equals_square}, we conclude that
$\mathbb P(T_k=1\mid H_1)\to 1$ and $\mathbb P(T_k=0\mid H_1)\to 0$.
\end{proof}
\iffalse
\begin{proof}[Proof of Corollary~\ref{cor:selected_window}]
Since \(k_0/k\to 1-c\), for any \(l\in\mathcal K_t\) with \(l\ge k\), we have
\[
\frac{l-k_0}{l}
=
1-\frac{k_0}{l}
\ge
1-\frac{k_0}{k}
=
\frac{k-k_0}{k}
\to c
\]
as \(t\to\infty\). Thus, every candidate window \(l\ge k\) contains a non-vanishing proportion of pre-break observations.

According to \eqref{con:reject_large_window}, similar to the proof of Theorem \ref{thm:reject_large_window}, we have
\[
\begin{aligned}
&\mathbb P\left(
\bigcap_{l\in\mathcal K_t:\,l\ge k}\{T_l=1\}
\,\middle|\,H_1
\right) \\
&\quad\ge
\mathbb P\left(
\inf_{l\in\mathcal K_t:\,l\ge k}
\left\{
f_{t,k_0}(\hat\mu_{t,l})
-
f_{t,k_0}(\hat\mu_{t,k_0})
\right\}
>
\tau(t,k_0)
\,\middle|\,H_1
\right)
\to 1.
\end{aligned}
\]
Since \(\hat k_t\) is the largest admissible window,
\[
\{\hat k_t<k\}
=
\bigcap_{l\in\mathcal K_t:\,l\ge k}\{T_l=1\}.
\]
Therefore,
\[
\mathbb P(\hat k_t<k\mid H_1)\to 1.
\]
\end{proof}
\fi
\begin{proof}[Proof of Corollary~\ref{cor:selected_window}]
For any \(l\in\mathcal K_t^+\), similar to Theorem \ref{thm:reject_large_window}, 
\[
\bar X_{t,l}
=
\frac{l-k_0}{l}\bar X_{1,l}
+
\frac{k_0}{l}\bar X_2,
\qquad
\bar X_{t,k_0}=\bar X_2,
\]
where
\(\bar X_{1,l}
=
\frac{1}{l-k_0}\sum_{r=t-l}^{t-k_0-1}X_r\)
and \(\bar X_2
=
\frac{1}{k_0}\sum_{r=t-k_0}^{t-1}X_r .
\)
Hence,
\[
f_{t,k_0}(\hat\mu_{t,l})-f_{t,k_0}(\hat\mu_{t,k_0})=(\bar X_{t,l}-\bar X_{t,k_0})^2
=
\frac{(l-k_0)^2}{l^2}(\bar X_{1,l}-\bar X_2)^2.
\]
For \(l\ge k\),
\[
\frac{l-k_0}{l}
=
1-\frac{k_0}{l}
\ge
1-\frac{k_0}{k}
=
\frac{k-k_0}{k}
\to c>0.
\]
We next show that the sample means converge uniformly over \(l\in\mathcal K_t^+\).
By Chebyshev's inequality,
\[
\mathbb P\left(
|\bar X_{1,l}-\mu_1|>\eta
\,\middle|\,H_1
\right)
\le
\frac{\operatorname{Var}_{\mathbb P_1}(X)}{\eta^2(l-k_0)},
\]
and using the union bound and \eqref{con:sparse},
\[\mathbb P\left(\sup_{l\in\mathcal K_t^+}
|\bar X_{1,l}-\mu_1|>\eta
\,\middle|\,H_1
\right)\le
\frac{r_t\operatorname{Var}_{\mathbb P_1}(X)}{\eta^2(k-k_0)}
\to0.\]
Therefore under $H_1$,
\(
\sup_{l\in\mathcal K_t^+}
|\bar X_{1,l}-\mu_1|
=o_p(1).
\)
Moreover, $\bar X_2=\mu_2+o_p(1)$, and by the triangle inequality,
\[
|\bar X_{1,l}-\bar X_2|
\ge
|\mu_1-\mu_2|
-
|\bar X_{1,l}-\mu_1|
-
|\bar X_2-\mu_2|.
\]
Since \((k-k_0)/k\to c\) and \(c>0\), the above uniform convergence implies that, for given \(\epsilon\) in \eqref{con:limit_tau},
\[
\mathbb P\left(
\inf_{l\in\mathcal K_t^+}
(\bar X_{t,l}-\bar X_{t,k_0})^2
>
c^2(\mu_1-\mu_2)^2-\frac{\epsilon}2
\,\middle|\,H_1
\right)
\to1.
\]
Combining it with \eqref{con:limit_tau} gives
\[
\mathbb P\left(
\inf_{l\in\mathcal K_t^+}
\left\{
f_{t,k_0}(\hat\mu_{t,l})
-
f_{t,k_0}(\hat\mu_{t,k_0})
\right\}
>
\tau(t,k_0)
\,\middle|\,H_1
\right)
\to1.
\]
Thus, with probability tending to one, every candidate window \(l\in\mathcal K_t^+\) is rejected by the pairwise comparison with the reference window \(k_0\), leading to the event \(\bigcap_{l\in\mathcal K_t:\,l\ge k}\{T_l=1\}\).

Since \(\hat k_t\) is the largest admissible window,
\(
\{\hat k_t<k\}
=
\bigcap_{l\in\mathcal K_t:\,l\ge k}\{T_l=1\}.
\)
Therefore,
\[
\mathbb P(\hat k_t<k\mid H_1)\to1,
\quad\text{or equivalently,}\quad
\mathbb P(\hat k_t\ge k\mid H_1)\to0.
\]
This completes the proof. 
\end{proof}
\section{Complementary analysis for error control}

The analysis in Section~\ref{subsec:iid_bootstrap} justifies the bootstrap threshold from an asymptotic \(M\)-estimation perspective. Complementing this argument, this appendix uses the closeness framework of~\cite{huang2025stability} to provide an  error-control interpretation for the  stability test. %In particular, we show that false rejection can occur only when the relevant closeness event or threshold event fails. 

For fixed \(t\) and two window lengths \(i<k\), the main object is the comparison statistic
\[
S_{t,i,k} := f_{t,i}(\hat{\boldsymbol{\theta}}_{t,k}) -
f_{t,i}(\hat{\boldsymbol{\theta}}_{t,i}), \qquad \hat{\boldsymbol{\theta}}_{t,m} \in \arg\min_{\boldsymbol{\theta}\in\Theta} f_{t,m}(\boldsymbol{\theta}), \quad m\in\{i,k\}.
\]
We interpret the comparison as a pairwise stability test under the local null hypothesis
\[H_0^{t,k}:~\text{no distributional shift occurs within a window}~ k.\]
Under \(H_0^{t,k}\), the window $k$ should not lead to a large excess loss on the reference window $i$.
For a window length $m$, define %the window-level population loss by
\[
F_{t,m}(\boldsymbol{\theta})
:=
\frac{1}{m}\sum_{l=t-m}^{t-1}F_l(\boldsymbol{\theta}),
\qquad
F_l(\boldsymbol{\theta})
:=
\mathbb{E}\!\left[\ell(\boldsymbol{X}_l,\boldsymbol{\theta})\right].
\]
Here $F_{t,m}$ represents the population loss averaged over the window
$\{t-m,\ldots,t-1\}$, while $F_t$ denotes the current target population loss.
%We use $\psi(t,m)$ to denote the stochastic error between the empirical loss $f_{t,m}$ and its population counterpart $F_{t,m}$, and use $\bar{\delta}(t,m)$ to denote the local nonstationarity error between $F_{t,m}$ and $F_t$.

Following~\cite{huang2025stability}, we use the following closeness relation between loss functions. For two lower-bounded functions $f,g:\Theta\to\mathbb{R}$, we say that $f$ and $g$ are
$(\epsilon,\delta)$-close if, for all
$\boldsymbol{\theta}\in\Theta$,
\[
g(\boldsymbol{\theta})-\inf_{\boldsymbol{\vartheta}\in\Theta}
g(\boldsymbol{\vartheta})
\le
e^{\epsilon}
\left\{
f(\boldsymbol{\theta})-\inf_{\boldsymbol{\vartheta}\in\Theta}
f(\boldsymbol{\vartheta})
+\delta
\right\},
\]
and
\[
f(\boldsymbol{\theta})-\inf_{\boldsymbol{\vartheta}\in\Theta}
f(\boldsymbol{\vartheta})
\le
e^{\epsilon}
\left\{
g(\boldsymbol{\theta})-\inf_{\boldsymbol{\vartheta}\in\Theta}
g(\boldsymbol{\vartheta})
+\delta
\right\}.
\]

We first recall a property of the closeness relation established by \cite{huang2025stability}.

\begin{lemma}
\label{lem:closeness-composition}
Suppose that $f$ and $g$ are $(\epsilon_1,\delta_1)$-close, and that
$g$ and $h$ are $(\epsilon_2,\delta_2)$-close. Then $f$ and $h$ are
$(\epsilon_1+\epsilon_2,\delta_1+\delta_2)$-close.
\end{lemma}

\begin{proof}
The result directly follows the definition; see \cite{huang2025stability} for details.
\end{proof}
We define the closeness event
\[
E_{t,k}:=
\bigcap_{\substack{i\in\mathcal K_t\\ i\le k}}
\left\{
f_{t,i}\text{ and }F_{t,i}\text{ are }
(\epsilon,\psi(t,i))\text{-close}
\right\}.
\]
and the threshold event
\[
A_{t,k}:=
\bigcap_{\substack{i\in\mathcal K_t\\ i<k}}
\left\{
\tau(t,i)\ge 2e^{2\epsilon}\psi(t,i)
\right\}.
\]
\iffalse
\[
E_{t}:=
\left\{
f_{t,m}\text{ and }F_{t,m}\text{ are }
(\varepsilon,\psi(t,m))\text{-close}~\text{for}~ m\in\{i, k\}
\right\}
\]
and the threshold event
\[A_{t}:=\{\tau(t,i)\ge 2e^{2\epsilon}\psi(t,i)\}.\]
\fi
\iffalse
Define
$B_{t,i,k}
:=
e^{4\epsilon}
\left\{\psi(t,k)+\bar{\delta}(t,k)\right\}
+
e^{2\epsilon}
\left\{\psi(t,i)+\bar{\delta}(t,i)\right\}$
\fi 

The following proposition gives a deterministic error-control
bound on the stability test. 
\begin{proposition}
Assume that the closeness error level satisfies \(\psi(t,k)\le \psi(t,i)\) for all \(i\in\mathcal K_t\) with \(i<k\) under \(H_0^{t,k}\). Then, $\bigcup_{i\in\mathcal K_t,~i<k}\{S_{t,i,k}>\tau(t,i)\}
\subseteq E_{t,k}^c\cup A_{t,k}^c.$
Consequently,
\[
\mathbb{P}
\left(
T_k=1\mid H_0^{t,k}
\right)
\le
\mathbb{P}
\left(
E_{t,k}^c\mid H_0^{t,k}
\right)
+\mathbb{P}
\left(
A_{t,k}^c\mid H_0^{t,k}
\right).
\]
\label{prop:2}
\end{proposition}
\begin{proof}%[Proof of Proposition~\ref{prop:error-control}]
For each \(i\le k\) with \(i\in\mathcal K_t\), on  \(E_{t,k}\),
\(f_{t,i}\) and \(F_{t,i}\) are \((\epsilon,\psi(t,i))\)-close. 
Under \(H_0^{t,k}\), no distributional shift occurs within the candidate window, and hence
\(F_{t,i}=F_t\) for all \(i\le k\). 
By Lemma~\ref{lem:closeness-composition}, under \(H_0^{t,k}\), \(f_{t,i}\) and
\(f_{t,k}\) are \((2\epsilon,\psi(t,i)+\psi(t,k))\)-close on \(E_{t,k}\). 

Using
\(
f_{t,k}(\hat{\boldsymbol{\theta}}_{t,k})
=
\inf_{\boldsymbol{\vartheta}\in\Theta} f_{t,k}(\boldsymbol{\vartheta})
\) and \(
f_{t,i}(\hat{\boldsymbol{\theta}}_{t,i})
=
\inf_{\boldsymbol{\vartheta}\in\Theta} f_{t,i}(\boldsymbol{\vartheta})
\), we have
\[
f_{t,i}(\hat{\boldsymbol{\theta}}_{t,k})
-
f_{t,i}(\hat{\boldsymbol{\theta}}_{t,i})
\le
e^{2\epsilon}
\left\{
\psi(t,i)+\psi(t,k)
\right\}
\]
on $E_{t,k}$ under $H_0^{t,k}$.
By the assumption \(\psi(t,k)\le \psi(t,i)\) for \(i<k\), we further have \[
f_{t,i}(\hat{\boldsymbol{\theta}}_{t,k})
-
f_{t,i}(\hat{\boldsymbol{\theta}}_{t,i})
\le
2e^{2\epsilon}
\psi(t,i).
\]
Therefore, on $E_{t,k}\cap A_{t,k}$,
\(
S_{t,i,k}
\le \tau(t,i)
\) for $i<k$ and $i\in\mathcal K_t$.
Consequently, under \(H_0^{t,k}\),
\[
\bigcup_{\substack{i\in\mathcal K_t\\ i<k}}\{S_{t,i,k}>\tau(t,i)\}
\subseteq
E_{t,k}^c\cup A_{t,k}^c.
\]
Since \(T_k=1\) if and only if at least one pairwise comparison rejects, the probability bound follows from the union event above.
\end{proof}
The closeness event \(E_{t,k}\) controls the empirical-population approximation error, while the threshold event \(A_{t,k}\) accounts for the randomness of the bootstrap thresholds. 
Proposition~\ref{prop:2} shows that, under \(H_0^{t,k}\), a false rejection of the stable candidate window \(k\) can occur only if  either the closeness  fails or the bootstrap threshold is too small. 
Thus, when the empirical losses are close to their population counterparts and the bootstrap thresholds dominate the corresponding stochastic error levels with high probability under the null, the stability test for window $k$ has a small type-I error probability.

\section{Supplementary algorithm, tables, and figures}
\subsection{Algorithm for online BAWS}
We assume that observations are collected sequentially up to time $T-1$, 
 and the forecasting procedure is conducted from an initial prediction time $t_0$ to  time $T$. The online BAWS algorithm is shown below.

 \begin{algorithm}[H]
\caption{Bootstrap-based adaptive window selection (BAWS)}\label{alg1}
\KwIn{Sequential data $\{\bm x_i\}_{i=1}^{T-1}$ and threshold level $\beta$.}
For  $t=t_0,\dots,T:$\\
\hspace{2mm} For $k\in\mathcal{K}_t$:\\
\hspace{5mm}Using samples $\{\bm x_i\}_{i=t-k}^{t-1}$, compute a minimizer $\hat{\bm\theta}_{t,k}$ of $f_{t,k}$, and \\ \hspace{5mm}calculate $\tau (t,k)$ via the bootstrap procedure;\\
\hspace{5mm} For all $i<k$ and $i\in\mathcal{K}_t$, judge if $T_{i,k}=0$; If so,  let $T_k=0$.\\
\hspace{2mm} Return  $\hat k_t=\max\{k\in\mathcal{K}_t:T_k=0\}$ and the estimator $\hat{\bm \theta}_{t}=\hat{\bm \theta}_{t,\hat k_t}$.\\
\KwOut{ $\{\hat k_t\}_{t=t_0}^T$ and $\{\hat{\bm\theta}_t\}_{t=t_0}^T$.}
\end{algorithm}
\subsection{Additional tables and figures}
\label{sec:supp-additional-tables}
\begin{table}[!ht]
\caption{The bias, variance, MSE, cumulative risk, and forecast loss for mean  across BAWS, SAWS, fixed windows (250, 500, 750), and full window across Settings A1--A3.}
\scriptsize
\renewcommand{\arraystretch}{0.8}
\resizebox{\textwidth}{!}{
\begin{tabular}{clcccccc}
\toprule
  \multirow{2}{*}{Mean} & & \multirow{2}{*}{BAWS} & \multirow{2}{*}{SAWS} &\multicolumn{3}{c}{Fixed Window} & \multirow{2}{*}{Full}\\
  \cline{5-7}
 &&&&250 & 500 & 750 \\

\midrule
 \multirow{4}{*}{A1} &MAB     & 0.0077 & 0.0083 & 0.0845 & 0.1678 & 0.2511 & 0.4631 \\
                       &Var & 0.0079 & 0.0073 & 0.0010 & 0.0005 & 0.0003 & 0.0002 \\
                       & CR       & 17.9276 & \textbf{14.8314} & 85.5601 & 168.4481 & 251.7759 & 501.8624 \\
                       & CL    & 392.9567   & \textbf{389.8877} & 460.3182 & 543.2237 & 626.4365 & 876.6553 \\
\cline{2-8}
  \multirow{4}{*}{A2} & MAB    & 0.0169 & 0.0123 & 0.2514 & 0.5010 & 0.7106 & 0.7744 \\
                       & Var & 0.0086 & 0.0088 & 0.0010 & 0.0005 & 0.0003 & 0.0002 \\
                       & CR       & 41.7347 & \textbf{22.6393} & 420.4758 & 836.3033 & 1188.3080 & 1201.3810 \\
                      &CL    & 417.3304& \textbf{397.9329} & 795.9474 & 1211.8626 & 1564.1570 & 1577.5416 \\
\cline{2-8}
  \multirow{4}{*}{A3} & MAB     & 0.0216 & 0.0133 & 0.2516 & 0.5010 & 0.7106 & 0.7744 \\
                       & Var & 0.0227 & 0.1941 & 0.0025 & 0.0012 & 0.0008 & 0.0004 \\
                      & CR       & \textbf{64.1937} & 297.7890 & 422.8770 & 837.5919 & 1189.3490 & 1201.7770 \\
                      & CL     & \textbf{1034.7530} & 1265.8910 & 1393.0080 & 1807.9140 & 2160.1410 & 2173.0750 \\
\bottomrule
\end{tabular}
}
\footnotesize
\emph{Note.} For mean forecasting, \(\mathrm{CR}=(T-t_0+1)\mathrm{MSE}\); hence, the MSE is omitted.
\label{tab:supp_Sim1}
\end{table}

\textbf{Auxiliary mean-forecasting results.}
Tables~\ref{tab:supp_Sim1} and ~\ref{tab:supp_Sim2} report the mean-forecasting results for the discrete-break settings in Section \ref{Sec:Sim1} and continuous-shift settings \ref{Sec:Sim2}, respectively. 
These results complement the VaR findings in the main text and illustrate the behavior of BAWS under another elicitable target. 
In the discrete-break settings, BAWS and SAWS generally reduce MAB, CR, and CL relative to fixed-window and full-window benchmarks. 
The advantage of BAWS is more evident in Setting A3, where the piecewise variance complicates the window selection and the bootstrap threshold adapts to changes in uncertainty. 
In the continuous-shift settings, BAWS performs well under cyclic fluctuations (B1) and persistent stochastic drift (B2), while its advantage becomes less pronounced than SAWS under Setting B3.
\begin{table}[!ht]
\caption{The bias, variance, MSE, cumulative risk, and forecast loss for mean across BAWS, SAWS, fixed windows (250, 500, 750), and full window across Settings B1--B3.}
\scriptsize
\renewcommand{\arraystretch}{0.8}
\centering\resizebox{0.98\textwidth}{!}{
\begin{tabular}{clccccccc}
\toprule
  \multirow{2}{*}{Mean} & & \multirow{2}{*}{BAWS} & \multirow{2}{*}{SAWS} &\multicolumn{3}{c}{Fixed Window} & \multirow{2}{*}{Full}\\
  \cline{5-7}
 &&&&250 & 500 & 750 \\
\midrule
\multirow{4}{*}{B1} &MAB    & 0.0488
            & 0.0782        & 0.2277             & 0.4104            & 0.5269             & 0.6790           \\
                                            & Var & 0.0120
         & 0.0099        & 0.0010             & 0.0005            & 0.0003             & 0.0002           \\
                                            & CR & \textbf{22.1111}
           & 25.9628       & 101.7474           & 330.5937          & 570.5037           & 952.2810         \\
                                             & CL    & \textbf{397.4562}            & 401.3663      & 477.2601           & 706.3232          & 946.2050           & 1327.7445        \\
\cline{2-8}
                       \multirow{4}{*}{B2} & MAB&   0.0538 
       & 0.0707        & 0.1368             & 0.1922            & 0.2306             & 0.3531           \\
                                             & Var & 0.0119 
 & 0.0099        & 0.0010             & 0.0005            & 0.0003             & 0.0002           \\
                                             & CR       & \textbf{24.3699} 
  & 25.8738       & 41.9695            & 74.8076           & 104.7016           & 255.1436         \\
                                            & CL     & \textbf{399.6568}          & 401.0672      & 417.1856           & 450.0374          & 479.8901           & 630.3935         \\
                       \cline{2-8}
                        \multirow{5}{*}{B3} & MAB   &0.0364& 	0.0469 	&0.0896&0.1289&0.1597& 	0.1817\\
                                          & Var &0.0100&0.0081&0.0010 	&0.0005&0.0003 &0.0002\\
                                             & CR& 17.9327&\textbf{17.2064}&20.0657& 	34.5248& 	53.1077& 	78.2319 \\
                                            & CL & 393.0637&\textbf{392.3147} 	&395.1584&409.7283&428.3650 	&453.5585 \\
                                            \bottomrule
\end{tabular}
} \label{tab:supp_Sim2}
\end{table}

\textbf{Additional simulation figures.}
Figures~\ref{Fig:A.4}--\ref{Fig:A.5} and Figures~\ref{Fig:B2}--\ref{Fig:B3} report additional trajectory plots for Settings A2--A3 and B2--B3, respectively.

\begin{figure}[t]
    \centering
 \includegraphics[width=\linewidth]{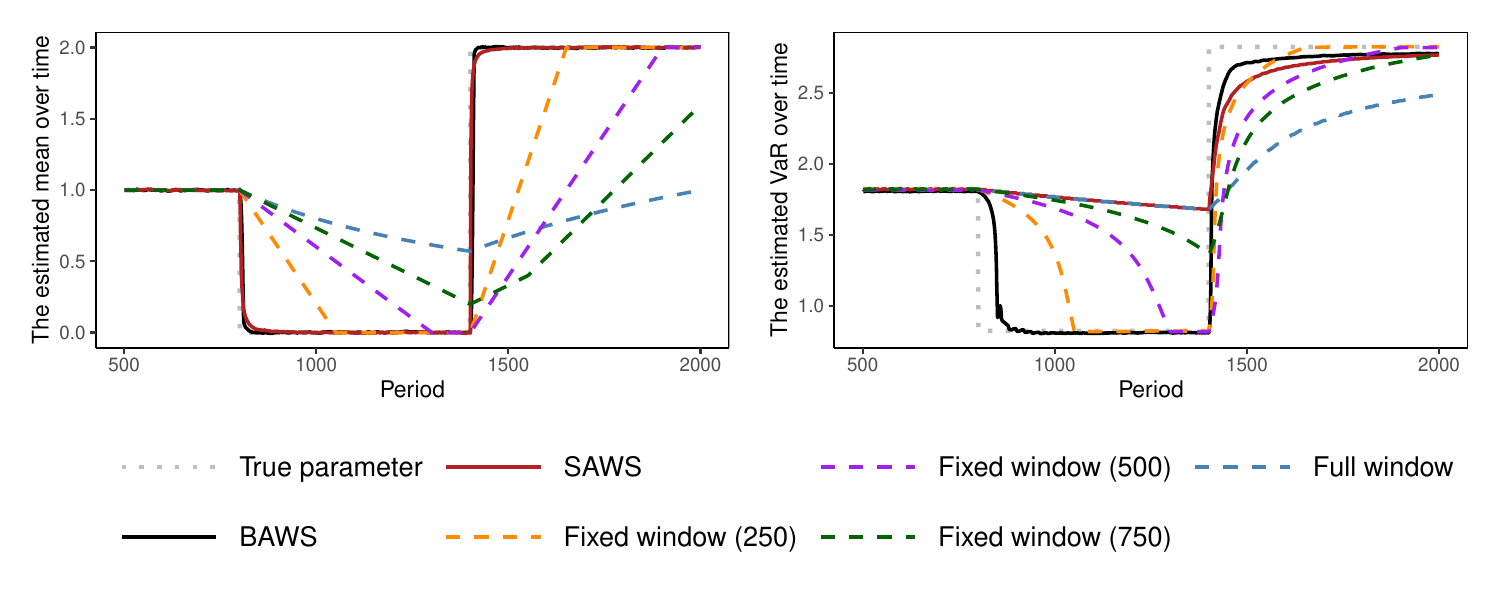}
    \caption{The patterns of mean and VaR estimators over time under Setting A2.}
    \label{Fig:A.4}
\end{figure}
\begin{figure}[t]
    \centering
 \includegraphics[width=\linewidth]{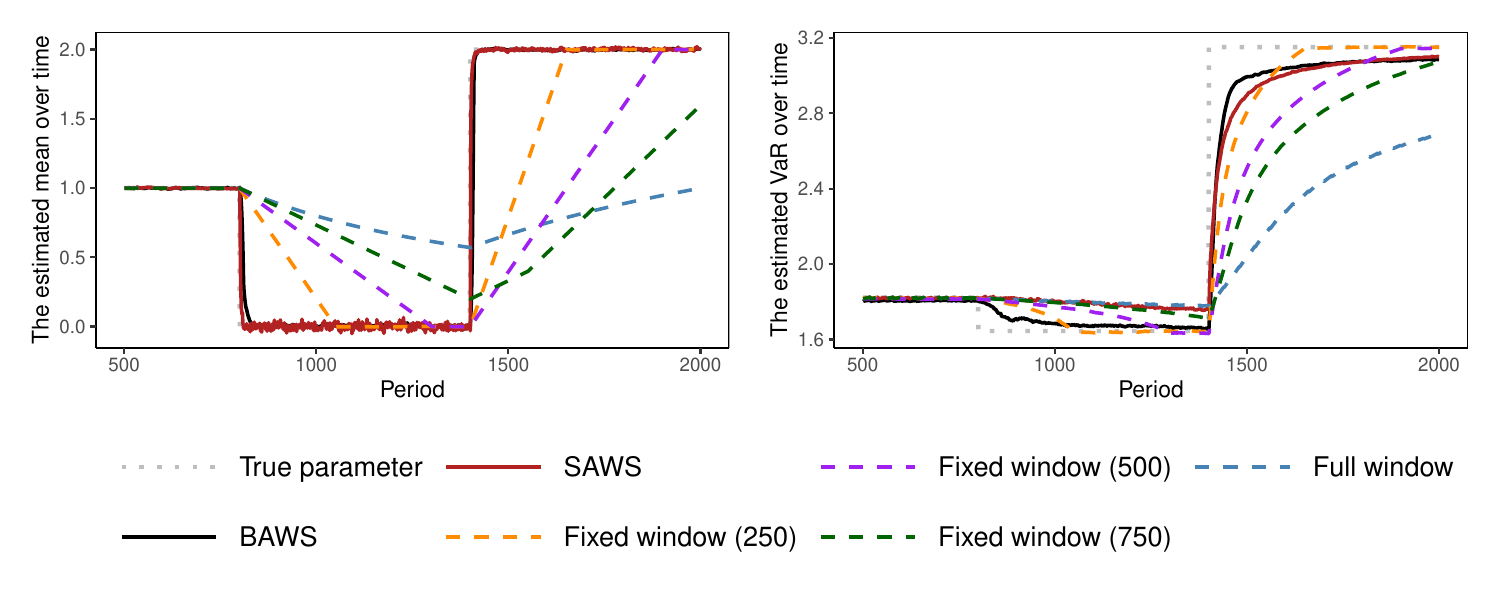}
    \caption{The patterns of mean and VaR estimators over time under Setting A3.}
    \label{Fig:A.5}
\end{figure}

\begin{figure}[!ht]
    \centering
 \includegraphics[width=0.98\linewidth]{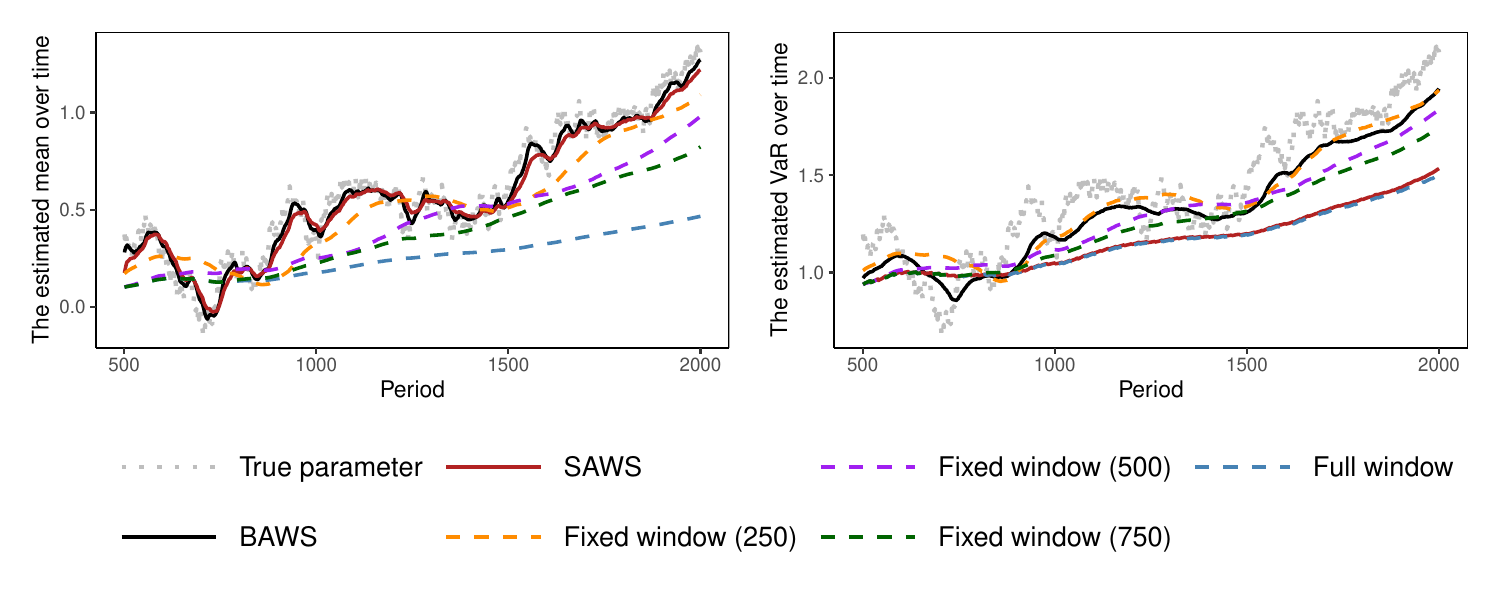}
    \caption{The patterns of mean and VaR estimators over time under Setting B2. }
    \label{Fig:B2}
\end{figure}

\begin{figure}[!ht]
    \centering
 \includegraphics[width=0.98\linewidth]{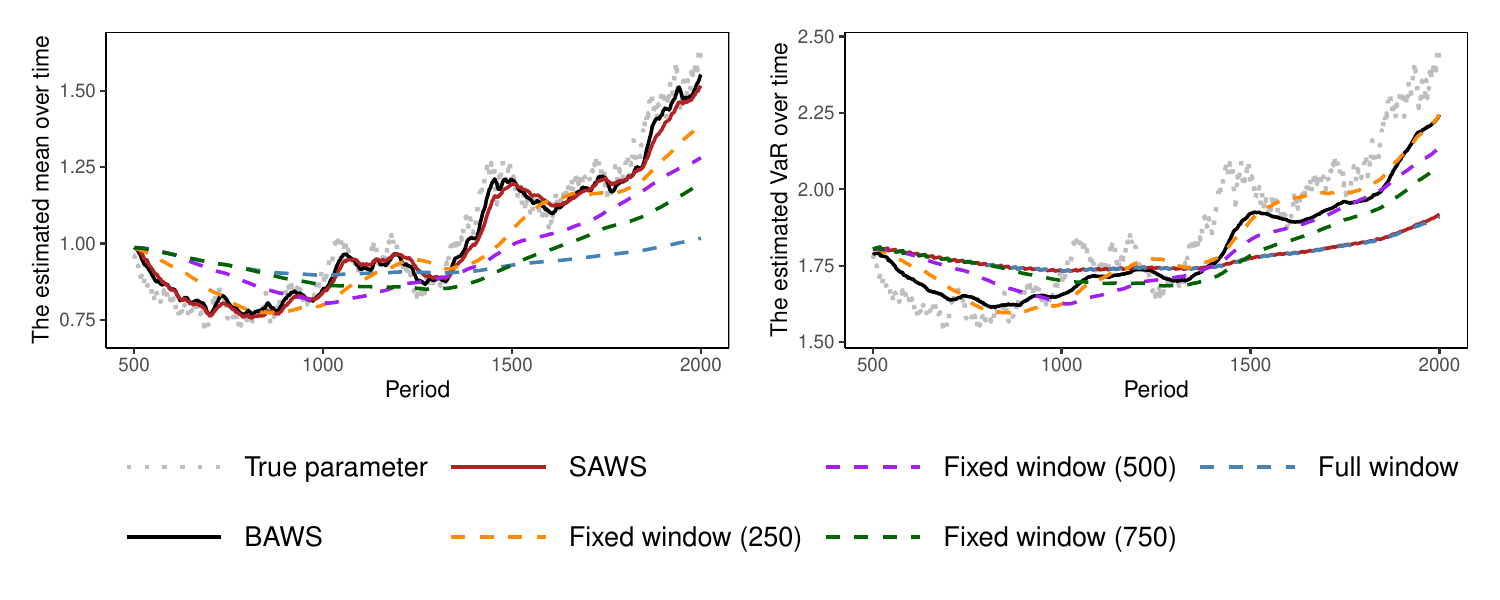}
    \caption{The patterns of mean and VaR estimators over time under Setting B3. }
    \label{Fig:B3}
\end{figure}

\begin{figure}[!t]
    \centering
        \includegraphics[width=0.7\linewidth, height=0.5\textheight]{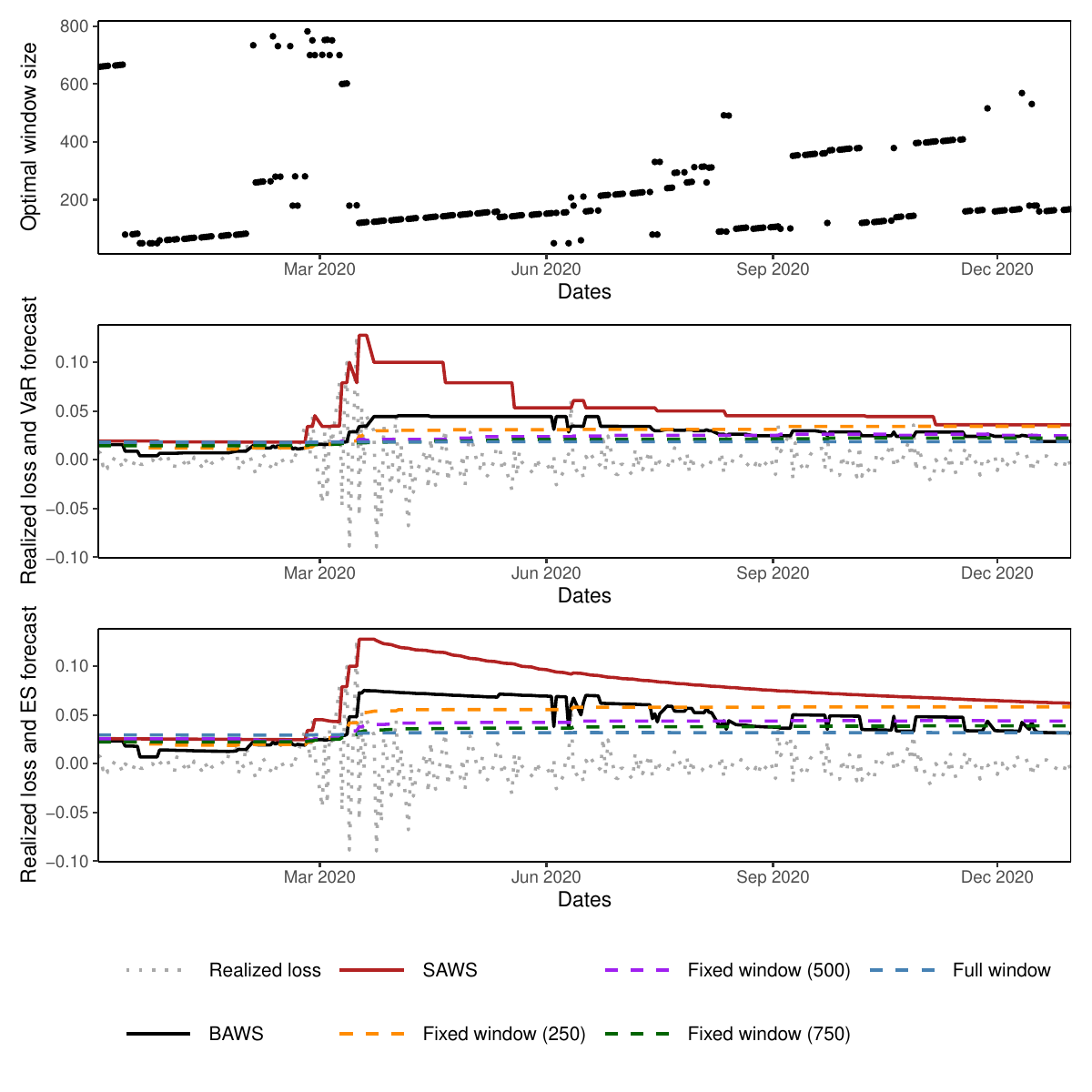}
    \caption{Temporal dynamics of optimal window sizes, VaR, and ES forecasts during COVID-19 pandemic. Top panel: Optimal window. Middle panel: VaR forecasts. Bottom panel: ES forecasts.}
    \label{Fig:real_2020}
\end{figure}

\begin{figure}[!t]
    \centering
        \includegraphics[width=0.7\linewidth, height=0.5\textheight]{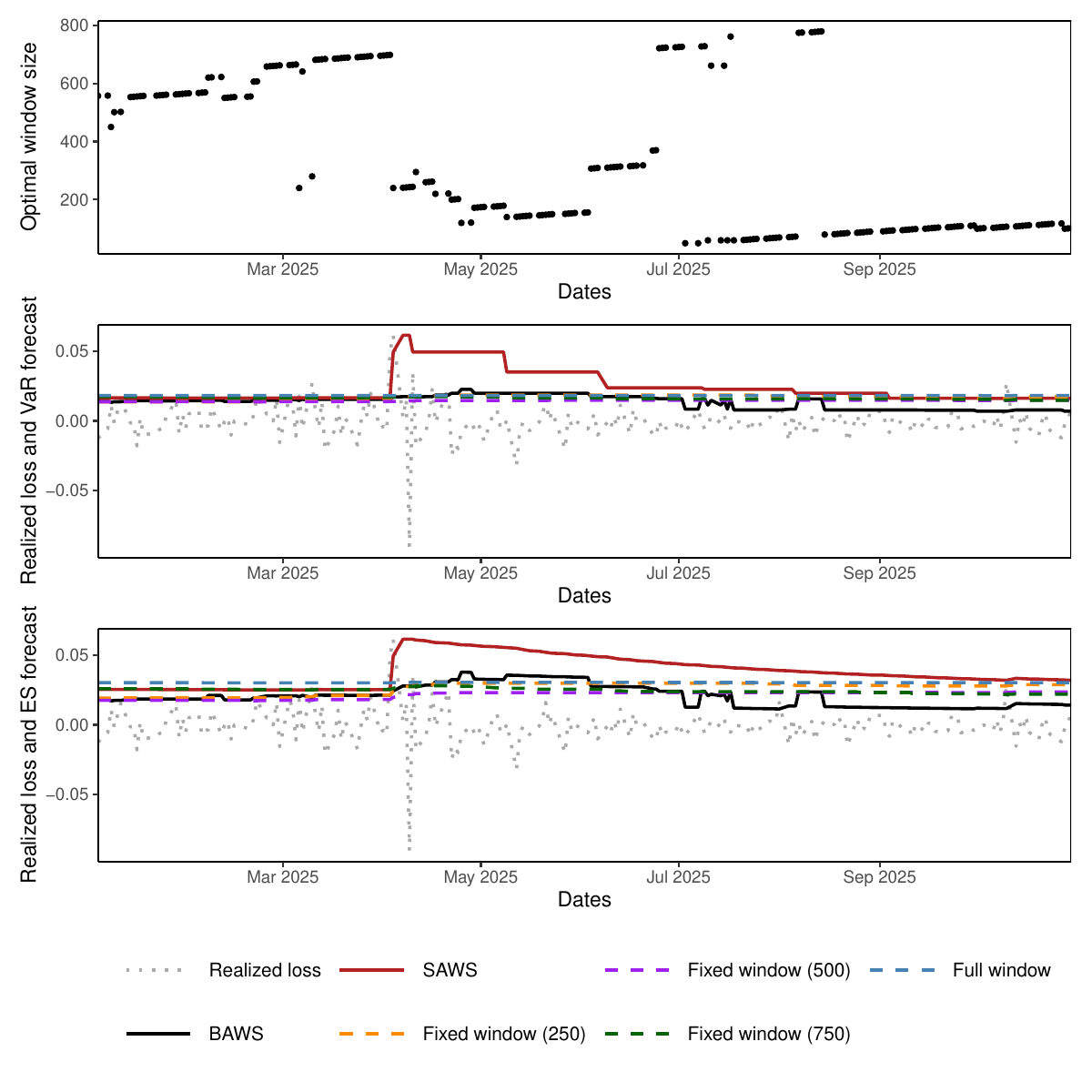}
    \caption{Temporal dynamics of optimal window sizes, VaR, and ES forecasts during the 2025 U.S. tariff measures. Top panel: Optimal window. Middle panel: VaR forecasts. Bottom panel: ES forecasts.}
    \label{Fig:real_2024}
\end{figure}

\clearpage
\subsection{Sensitivity analysis of BAWS threshold levels}
\label{sec:supp-sensitivity}

We examine the sensitivity of BAWS to the bootstrap threshold level $\beta$ under Setting A1 in Section \ref{Sec:Sim1}. The main simulations use $\beta=0.90$. 
We consider $\beta\in\{0.80,0.85,0.90,0.95,0.99\}$ while keeping all other simulation configurations unchanged. Tables~\ref{tab:sensitivity-beta-mean} and~\ref{tab:sensitivity-beta-var} report the numerical results for mean and VaR forecasting, respectively. Figures~\ref{fig:sensitivity-beta-mean-global} and~\ref{fig:sensitivity-beta-var-global} provide an overview of the forecast paths around the structural break, while Figures~\ref{fig:sensitivity-beta-mean-path} and~\ref{fig:sensitivity-beta-var-path} zoom in on the post-break period. Overall, BAWS remains stable across threshold levels. Larger values of $\beta$ select longer windows and reduce variance, whereas smaller values of $\beta$ adapt more aggressively after the structural break. The global plots show that all threshold levels adjust to the new regime after the break, and the zoomed-in plots further show that the resulting forecast paths are close to each other after the initial adjustment period. The baseline choice $\beta=0.90$ therefore provides a balanced specification.

\begin{table}[!htbp]
\centering
\caption{Sensitivity of BAWS to the bootstrap threshold level $\beta$ under Setting A1 for mean forecasting.}
\label{tab:sensitivity-beta-mean}
\begin{tabular}{ccccccc}
\toprule
$\beta$ & MAB & Var & CR & CL & Avg. window & Median window \\
\midrule
0.80 & 0.0073 & 0.0130 & 24.8808 & 400.0473 & 263.0 & 127 \\
0.85 & 0.0074 & 0.0106 & 21.5663 & 396.6493 & 322.6 & 212 \\
0.90 & 0.0077 & 0.0079 & 17.9276 & 392.9567 & 390.3 & 334 \\
0.95 & 0.0092 & 0.0051 & 14.3023 & 389.2681 & 468.3 & 501 \\
0.99 & 0.0139 & 0.0027 & 11.9073 & 386.8051 & 547.0 & 590 \\
\bottomrule
\end{tabular}

\vspace{1mm}
\begin{minipage}{0.92\textwidth}
\footnotesize
\textit{Note.} 
``Avg. window'' and ``Median window'' denote the average and median selected window sizes over the forecasting period and Monte Carlo replications.
\end{minipage}
\end{table}

\begin{table}[!htbp]
\centering
\caption{Sensitivity of BAWS to the bootstrap threshold level $\beta$ under Setting A1 for VaR forecasting.}
\label{tab:sensitivity-beta-var}
\begin{tabular}{cccccccc}
\toprule
$\beta$ & MAB & Var & MSE & CR & CL & Avg. window & Median window \\
\midrule
0.80 & 0.0499 & 0.0084 & 0.0191 & 4.7433 & 82.0145 & 442.3 & 391 \\
0.85 & 0.0533 & 0.0069 & 0.0187 & 4.7723 & 82.0280 & 504.7 & 503 \\
0.90 & 0.0597 & 0.0055 & 0.0191 & 5.0391 & 82.3192 & 573.1 & 585 \\
0.95 & 0.0708 & 0.0042 & 0.0209 & 5.5876 & 82.8822 & 645.0 & 657 \\
0.99 & 0.0933 & 0.0030 & 0.0275 & 7.3334 & 84.6215 & 731.5 & 750 \\
\bottomrule
\end{tabular}

\vspace{1mm}
\begin{minipage}{0.92\textwidth}
\end{minipage}
\end{table}

\begin{figure}[!htbp]
\centering
\includegraphics[width=0.78\textwidth]{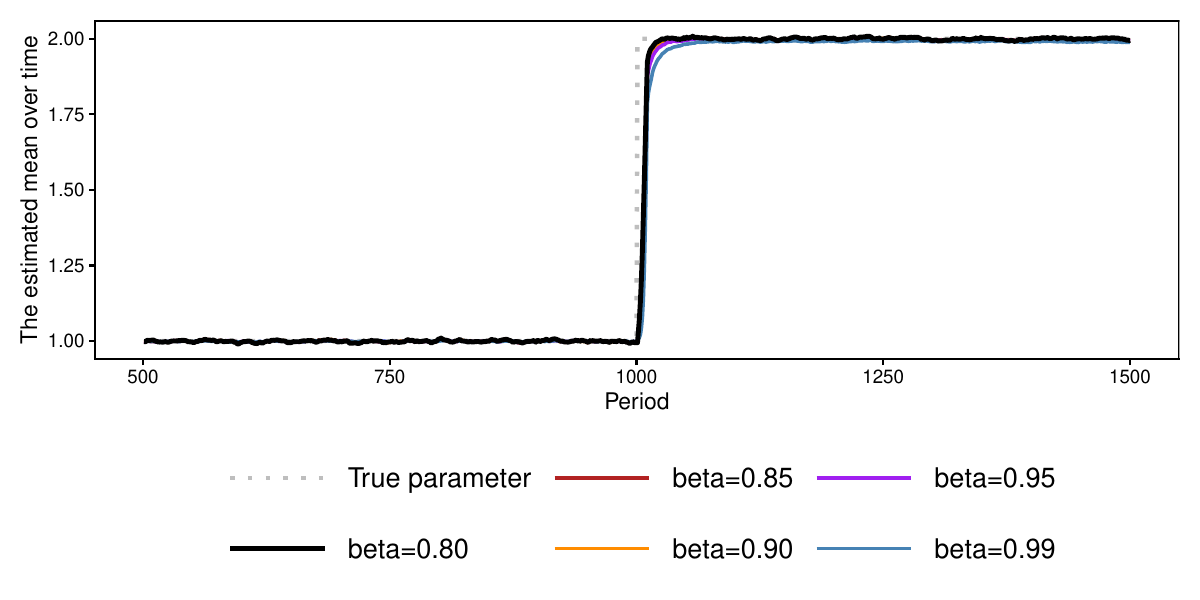}
\caption{Sensitivity of BAWS mean estimates to the bootstrap threshold level $\beta$ under Setting A1. The plot provides a global view over periods 500--1500, covering both the pre-break regime and the post-break adjustment period.}
\label{fig:sensitivity-beta-mean-global}
\end{figure}

\begin{figure}[!htbp]
\centering
\includegraphics[width=0.78\textwidth]{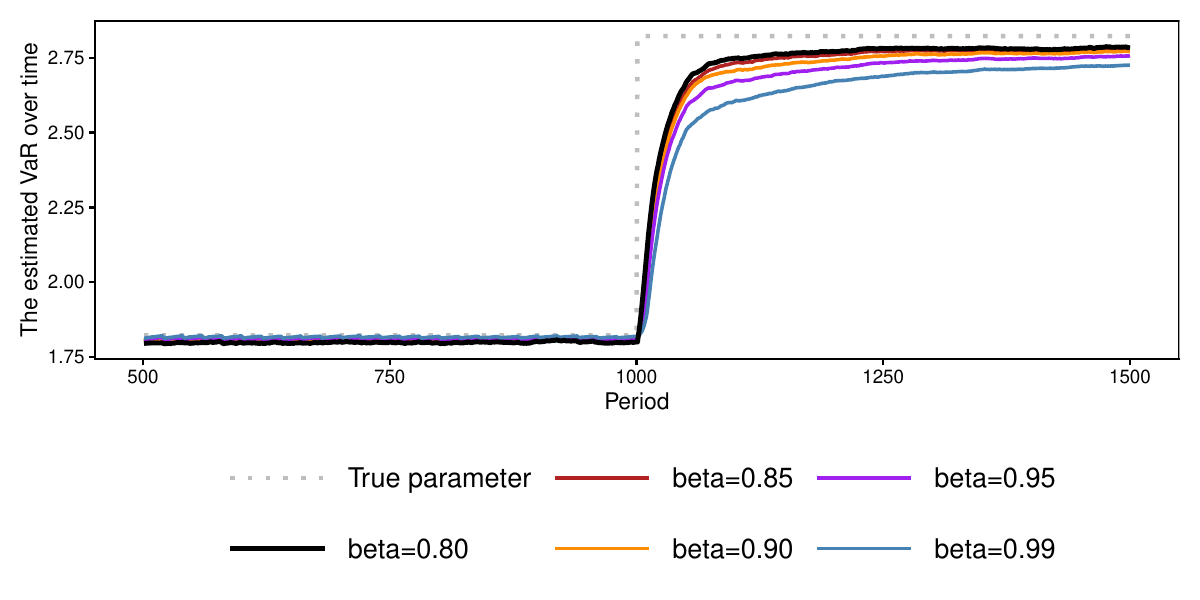}
\caption{Sensitivity of BAWS VaR estimates to the bootstrap threshold level $\beta$ under Setting A1. The plot provides a global view over periods 500--1500, covering both the pre-break regime and the post-break adjustment period.}
\label{fig:sensitivity-beta-var-global}
\end{figure}

\begin{figure}[!htbp]
\centering
\includegraphics[width=0.78\textwidth]{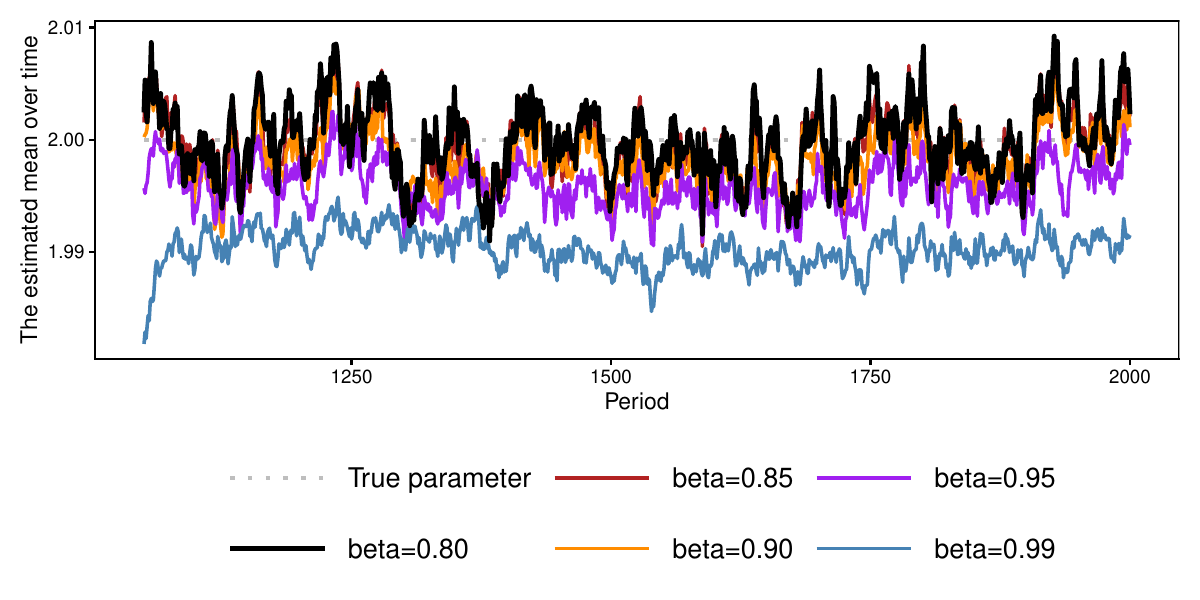}
\caption{Sensitivity of BAWS mean estimates to the bootstrap threshold level $\beta$ under Setting A1. The plot zooms in on periods 1050--2000 after the structural break. The true mean is equal to 2 in this period.}
\label{fig:sensitivity-beta-mean-path}
\end{figure}

\begin{figure}[!htbp]
\centering
\includegraphics[width=0.78\textwidth]{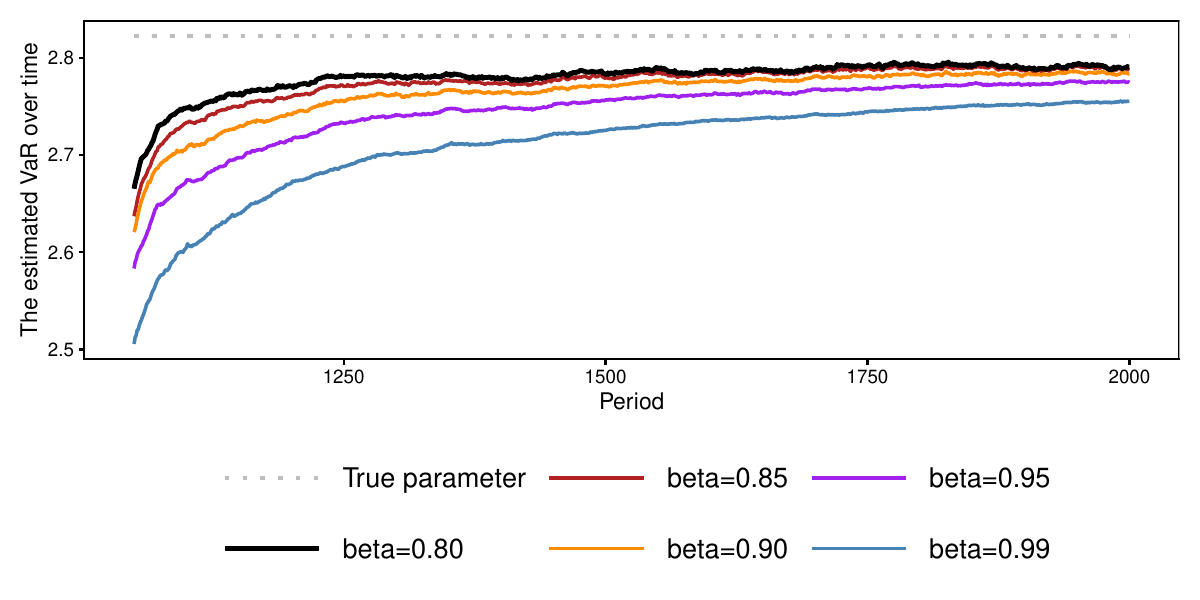}
\caption{Sensitivity of BAWS VaR estimates to the bootstrap threshold level $\beta$ under Setting A1. The plot zooms in on periods 1050--2000 after the structural break. The true VaR is equal to approximately 2.82 in this period.}
\label{fig:sensitivity-beta-var-path}
\end{figure}
\subsection{Sensitivity analysis of SAWS tuning parameters}\label{sec:supp-saws-sensitivity}
Table~\ref{tab:saws_sensitivity} reports a sensitivity analysis for the SAWS tuning
parameters in the empirical study. Following the threshold specification in \citet{huang2025stability}, we use the threshold for Lipschitz population losses,
\[
\tau(n,k)=C_\tau\sqrt{\log(\alpha_\tau^{-1}+1+n)/k},\]
and vary \(C_\tau\) and \(\alpha_\tau\) around the values used in the main analysis.  Compared with the BAWS results, some SAWS tuning choices yield lower average forecast
losses during the GFC and COVID periods, whereas all SAWS specifications considered here
have higher average forecast losses than BAWS over the full prediction period and during
the Tariff period.  Panel A fixes \(\alpha_\tau=0.1\) and varies \(C_\tau\), while
Panel B fixes \(C_\tau=0.05\) and varies \(\alpha_\tau\). The average loss is nearly
unchanged over a wide range of \(\alpha_\tau\), whereas very small or large values of
\(C_\tau\) can lead to larger forecast losses. Large values of \(C_\tau\), such as \(C_\tau=0.7\), produce large thresholds, so SAWS tends
to retain long windows and may behave similarly to the full-window benchmark. The parameter choice used in the main
analysis,  \(C_\tau=0.05\) and \(\alpha_\tau=0.1\), performs well across the full sample and
the main stress periods. Overall, the performance of SAWS varies across tuning
parameter specifications, whereas BAWS provides a data-driven alternative for threshold
calibration through the single threshold level \(\beta\).

\begin{table}[H]
\centering
\caption{Sensitivity of SAWS average forecast loss to tuning parameters.}
\label{tab:saws_sensitivity}
\resizebox{0.9\textwidth}{!}{
\begin{tabular}{lrrrrrrrr}
\toprule
\multicolumn{8}{l}{Panel A: fixed \(\alpha_\tau=0.1\), varying \(C_\tau\)}\\
\midrule
Period & BAWS&\(0.01\) & \(0.03\) & \(0.05\) & \(0.10\) & \(0.30\) & \(0.50\) & \(0.70\)\\
\midrule
2006--2025 &\textbf{2.2642} &
 3.1345 & 2.3490 &{2.3145} & 2.4093 & 2.4788 & 2.4447 & 2.4447\\
GFC        &3.2513 & 6.9525 & 3.5767 & \textbf{3.1852} & 3.5152 & 3.9584 & 4.1779 & 4.1779\\
COVID      &4.0501 & 6.9319 & 4.3460 & \textbf{3.8752} & 4.0408 & 4.1457 & 4.1412 & 4.1412\\
Tariff     &\textbf{2.3124} & 3.0044 & 2.4079 & {2.3856} & 2.3994 & 2.3354 & 2.3292 & 2.3292\\
\midrule
\multicolumn{8}{l}{Panel B: fixed \(C_\tau=0.05\), varying \(\alpha_\tau\)}\\
\midrule
Period &BAWS& \(0.01\) & \(0.03\) & \(0.05\) & \(0.10\) & \(0.30\) & \(0.50\) & \(0.70\)\\
\midrule
2006--2025 &\textbf{2.2642}& 2.3166 & 2.3166 & 2.3145 & {2.3145} & 2.3145 & 2.3145 & 2.3145\\
GFC &3.2513        & 3.2012 & 3.2012 & 3.1852 & \textbf{3.1852} & 3.1852 & 3.1852 & 3.1852\\
COVID&4.0501      & 3.8752 & 3.8752 & 3.8752 & \textbf{3.8752} & 3.8752 & 3.8752 & 3.8752\\
Tariff&\textbf{2.3124}     & 2.3856 & 2.3856 & 2.3856 & {2.3856} & 2.3856 & 2.3856 & 2.3856\\
\bottomrule
\end{tabular}
}\\
\vspace{1mm}
\begin{minipage}{0.9\textwidth}
\footnotesize
\textit{Note.} 
 Entries are average forecast losses in percentage.  The parameters used in the main empirical analysis are \(\alpha_\tau=0.1\) and
\(C_\tau=0.05\).
\end{minipage}
\end{table}
\clearpage

%\printbibliography[title={References}]

%\bibliographystyle{apalike}
%\bibliography{JASA/reference}

\end{document}